\documentclass[conference]{IEEEtran}

\usepackage{amsmath,amsthm,amssymb}
\usepackage{stmaryrd}
\usepackage{mathabx}
\usepackage{comment}
\usepackage{proof}

\mathchardef\mhyph="2D 

\newcommand{\interp}[1]{\llbracket #1 \rrbracket}
\newcommand{\interpl}[0]{\llbracket}
\newcommand{\interpr}[0]{\rrbracket}

\newcommand{\ctxtup}[1]{\ulcorner #1 \urcorner}
\newcommand{\ctxtdn}[1]{\llcorner #1 \lrcorner}
\newcommand{\ctxtmv}[1]{\langle #1 \rangle}

\newcommand{\case}[1]{\noindent\underline{Case #1:}}
\newcommand{\caseb}[0]{\noindent\underline{Case:}}

\newcommand{\pcase}[1]{\noindent\underline{#1:}}

\newcommand{\lam}[2]{\lambda\, #1.\, #2}
\newcommand{\all}[2]{\forall\, #1.\, #2}
\newcommand{\Lam}[2]{\Lambda\, #1.\, #2}
\newcommand{\thereis}[2]{\exists\, #1.\, #2}
\newcommand{\rec}[2]{\textit{rec}\, #1.\, #2}

\newcommand{\convl}[0]{\blacktriangleleft}
\newcommand{\convr}[0]{\blacktriangleright}

\newcommand{\dcdot}[0]{\mathbin{\ast}}
\newcommand{\To}[0]{\Rightarrow}
\newcommand{\releq}[0]{\topdoteq}

\newcommand{\rref}[1]{\{\ref{#1}\}}

\newtheorem{theorem}{Theorem}
\newtheorem{lemma}[theorem]{Lemma}
\newtheorem{corollary}[theorem]{Corollary}
\newtheorem{definition}[theorem]{Definition}

\newtheorem{proposition}[theorem]{Proposition}

\newcommand{\dparam}[0]{D_\textit{param}}
\newcommand{\dind}[0]{D_\textit{ind}}

\begin{document}

\def\vers{proofs}

\ifx\vers\undefined
  \excludecomment{opt}
  \title{Relational Type Theory}
\else
  \includecomment{opt}
  \title{Relational Type Theory (All Proofs)}

\author{\IEEEauthorblockN{Aaron Stump}
\IEEEauthorblockA{Computer Science\\
The University of Iowa\\
Iowa City, Iowa, 52242 \\
Email: aaron-stump@uiowa.edu}
\and
\IEEEauthorblockN{Benjamin Delaware}
\IEEEauthorblockA{Computer Science\\
Purdue University\\
West Lafayette, Indiana, 47907 \\
Email: bendy@purdue.edu}
\and
\IEEEauthorblockN{Christopher Jenkins}
\IEEEauthorblockA{Computer Science\\
The University of Iowa\\
Iowa City, Iowa, 52242 \\
Email: christopher-jenkins@uiowa.edu}
}

\fi

\maketitle

\begin{abstract}
  This paper introduces Relational Type Theory (RelTT), a new approach
  to type theory with extensionality principles, based on a relational semantics for
  types.  The type constructs of the theory are those of System F plus
  relational composition, converse, and promotion of application of a
  term to a relation.  A concise realizability semantics is presented
  for these types.  The paper shows how a number of constructions of
  traditional interest in type theory are possible in RelTT, including
  $\eta$-laws for basic types, inductive types with their induction
  principles, and positive-recursive types.  A crucial role is played
  by a lemma called Identity Inclusion, which refines the Identity
  Extension property familiar from the semantics of parametric
  polymorphism.  The paper concludes with a type system for RelTT,
  paving the way for implementation.
\end{abstract}


%
\IEEEpeerreviewmaketitle

\section{Introduction}

Modern constructive type theories have long wish lists of features,
from inductive and coinductive types, to type-specific extensionality
principles, quotient types, higher-order datatypes, and more.  In
tension with this, there are excellent reasons to seek to keep the
core type theory small and trustworty.  This has been done, in
different ways, for Lean~\cite{lean} and recently
Coq~\cite{sozeau+20}.  Both those systems implement (variants of) the
Calculus of Inductive Constructions, which lacks type-specific
extensionality principles.

The present paper proposes Relational Type Theory (RelTT) for deriving
expressive type constructs, with type-specific extensionality
principles, from a formally small core theory.  The approach followed
is, to the authors' knowledge, novel: RelTT is based a semantics for
types as binary relations on untyped terms.  For example, the semantics
for a function type $R\to R'$ makes it the set of pairs of terms
$(t_1,t_2)$ that jointly map inputs related by the meaning of $R$ to
outputs related by the meaning of $R'$ (the semantics familiar from
the field of logical relations).  The notion that a term ``is'' a
function is expressed only by saying that it is related to itself at
function type.  So relations between terms are the primary concern of
the theory, and expression of program behavior in isolation (i.e.,
traditional typing) is secondary.

This commitment to relational semantics leads us in an unexplored
direction: we extend the set of type constructs with relational
constructs.  We may use these to express asymmetric relations, which
are crucial for developing reasoning principles, like induction
principles, within the theory.  Interestingly, dependent types are
unnecessary for this.  The relational semantics already gives us a
form of dependency which is all we need for inductive reasoning about
terms.  So RelTT is an extension, with relational type constructs, of
System F, not the Calculus of Constructions.  Avoiding dependent types
notably simplifies the semantics.  The power of System F is needed
because the terms that the theory (relationally) types are those of
pure lambda calculus, so we adopt impredicative lambda encodings to
represent inductive types.

The contributions of the paper are:
\begin{itemize}
\item The syntax and semantics of relational types (Section~\ref{sec:syn})
\item Basic properties of the semantics, crucially including 
  $\beta\eta$-closure (Section~\ref{sec:basic})
\item Interesting derived type forms and examples (Section~\ref{sec:ex}), and basic type-specific extensionality principles
   (Section~\ref{sec:ext}).
\item Classes of types whose interpretations are proved to be, respectively, symmetric (Section~\ref{sec:symm}) and transitive (Section~\ref{sec:trans}).  The proof of transitivity crucially relies on a novel theorem dubbed Identity Inclusion (Lemma~\ref{thm:posneg}).  The intricate
  proof of this makes use of duality between types where all quantifiers occur only positively ($\forall^+$ types),
  and ones where they occur only negatively.
\item Derivation of induction principles from types for
  Church-encodings (Section~\ref{sec:indtp}).  This covers any
  inductive type definable by a type scheme which is positive and, due
  to a critical use of Identity Inclusion, $\forall^+$.
  Positive-recursive types are also derived (Section~\ref{sec:rectp}).
\item A proof system called RelPf (Section~\ref{sec:relpf}) and type system RelTy (Section~\ref{sec:relty}), which are proven
  sound with respect to the semantics, and are intended as the starting point for implementation of RelTT as a proof assistant.
  \end{itemize}

\noindent We will reference lemmas and theorems by name, with the theorem number
following in braces (e.g.,
``$\beta\eta$-Closure~\rref{lem:closedr}'').  We will label
assumptions and proven local facts with numbers (e.g., ``(1)''), and
goals with capital letters (e.g., ``(A)'').  
\ifx\vers\undefined
\footnote{Note to reviewers: proofs of elided lemmas are available in an arXiv report (submitted to arXiv prior to submission of this paper), which however is not anonymized; it could be consulted, if need be, following deanonymization.  A few statements labeled
as ``propositions'' are not proved.}
\else

\fi

\section{Relational types and their semantics}
\label{sec:syn}

\begin{figure}
  \[
  \begin{array}{lll}
   \textit{terms }t & ::= & x\ |\ \lam{x}{t}\ |\ t\, t' \\
   \textit{types }R & ::= & X\ |\ R \to R' \ |\ \all{X}{R}\ |\ R^{\cup}\ |\ R \cdot R' \ |\ t
  \end{array}
\]
  \caption{Syntax for relational types ($X$ ranges over type variables)}
  \label{fig:reltypes}
\end{figure}

The syntax of relational types $R$ is given in
Figure~\ref{fig:reltypes}.  Terms $t$ are those of pure untyped lambda
calculus.  Relational type constructs include those of System F, plus
$R^{\cup}$ for converse of a relation, $R \cdot R'$ for composition of
relations, and promotion of terms $t$ to relations, to be explained
shortly. Usual parsing precedences from type theory are followed;
additionally, $R^{\cup}$ binds most tightly, $R \cdot R'$ second most
tightly, and the other constructs after these.  We also follow the
usual convention that distinct meta-variables ranging over variables
denote distinct variables (so $x$ and $y$ denote different variables),
and treat terms and types up to $\alpha$-equivalence.
Capture-avoiding substitution of $R$ for $X$ in $R'$ is denoted
$[R/X]R'$ (similarly $[t/x]t'$ for terms).  The set of free variables
of any syntactic entity $e$ is denoted $\textit{FV}(e)$.  The obvious
definitions of these syntactic notions are omitted.

\begin{definition}
  A relation $r$ on terms of pure untyped $\lambda$-calculus is
  $\beta\eta$-closed iff $t_1\ [r]\ t_2$, $t_1'=_{\beta\eta} t_1$, and
  $t_2' =_{\beta\eta} t_2$ imply $t_1'\ [r]\ t_2'$.  Write
  $\mathcal{R}$ for the set of all such relations, and use meta-variable $r$ to range over $\mathcal{R}$
  \end{definition}

The relational semantics of types is defined in
Figure~\ref{fig:relsem}, where environment $\gamma$ is a function
mapping a finite set of type variables to elements of $\mathcal{R}$.
We use infix notation for application of a relation to a pair of terms
and, following~\cite{hermida14}, we sometimes put square brackets
around the relation for readability; for example, in the three
equations at the bottom of the figure.  In those equations, operators
like ``$\to$'' on the right-hand sides have their standard meaning in
the background meta-logic.

The interpretation $\interp{R}_\gamma$ is defined iff $\gamma$ is
defined for all free type variables of $R$.  When referencing
$\interp{R}_\gamma$ in theorems, we assume it is defined.  The
semantics extends $\gamma$ from type variables to arbitrary
types. Promotion of term $t$ is the graph of the meta-level operation
mapping $t'$ to $t\ t'$.  Many examples are below.

\begin{figure}
\[
\begin{array}{lll}
\interp{X}_\gamma & = & \gamma(X) \\
\interp{R \to R'}_\gamma & = & \interp{R}_\gamma \to \interp{R'}_\gamma\\
\interp{\all{X}{R}}_\gamma & = &  \bigcap_{r \in\mathcal{R}}\, \interp{R}_{\gamma[X\mapsto r]}\\
\interp{R^{\cup}}_\gamma & = & \interp{R}_\gamma^\cup \\
\interp{R \cdot R'}_\gamma & = & \interp{R}_\gamma \cdot \interp{R'}_\gamma \\
\interp{\hat{t}}_\gamma & = & \{(t,t')\ |\ \hat{t}\,t\,=_{\beta\eta}\,t'\} \\ \\

\multicolumn{3}{l}{\textit{where:}}\\
  t\, [r_1 \to r_2]\, t' & = & \all{a}{\all{a'}{a\,[r_1]\,a' \to t\,a\,[r_2]\,t'\,a'}} \\
  t\, [r^\cup]\, t' & = & t'\,[r]\, t\\
  t\, [r_1\cdot r_2]\, t' & = & \thereis{t''}{t\,[r_1]\,t''\,\wedge\,t''\,[r_2]\,t'}
  
\end{array}
\]
\caption{Semantics for relational types; relational operators $\to$, ${}^\cup$, and $\cdot$}
\label{fig:relsem}
\end{figure}

While the semantics of universal types quantifies (at the meta-level)
over all relations in $\mathcal{R}$, we will restrict ourselves in all
examples below to instantiating such quantifiers only with definable
relations (i.e., ones of the form $\interp{R}_\gamma$).  In
Sections~\ref{sec:relpf} and~\ref{sec:relty} below, we will consider deductive
systems for RelTT where this restriction will be enforced.

\section{Basic properties}
\label{sec:basic}

\begin{lemma}[$\beta\eta$-Closure]
\label{lem:closedr}
  $\interp{R}_\gamma\in\mathcal{R}$.
\end{lemma}
\begin{proof}
  The proof is by induction on $R$.  Suppose $t_1 =_{\beta\eta} t_1'$ and $t_2 =_{\beta\eta} t_2'$,
  assume (1) $t_1\ \interp{R}_\gamma\ t_2$, and show $t_1'\ \interp{R}_\gamma\ t_2'$.

  \case{$X$} $\gamma(X)\in\mathcal{R}$ by specification of $\gamma$.

  \case{$R\to R'$} assume (2) $a\ \interp{R}_\gamma\ a'$, and show $t_1'\ a\ \interp{R'}_\gamma\ t_2'\ a'$.  From (1) and (2) we
  have $t_1\ a\ \interp{R'}_\gamma\ t_2\ a'$.  From this, the IH gives us the required conclusion, as $t_1\ a =_{\beta\eta} t_1'\ a$
  and $t_2\ a' =_{\beta\eta} t_2'\ a'$.

  \case{$\all{X}{R}$} assume $r\in\mathcal{R}$, and show $t_1'\ \interp{R}_{\gamma[X\mapsto r]}\ t_2'$.  By (1), we have
  $t_1\ \interp{R}_{\gamma[X\mapsto r]}\ t_2$, from which the IH then yields the desired conclusion.

  \case{$R^{\cup}$} this follows by the IH (using symmetry of $=_{\beta\eta}$).

  \case{$R\cdot R'$} from (1), there exists $t$ such that $t_1\ \interp{R}_\gamma\ t$ and $t\ \interp{R'}_\gamma\ t_2$.
  By the IH, $t_1'\ \interp{R}_\gamma\ t$ and $t\ \interp{R'}_\gamma\ t_2'$.  These imply the desired
  conclusion.

  \case{$t$} from (1), we have $t\ t_1 =_{\beta\eta} t_2$;
   $t\ t_1' =_{\beta\eta} t_2'$ then follows.

\end{proof}

\begin{lemma}[Symmetry Properties]
\label{lem:sym}
  \
  
  \begin{enumerate}
  \item $\interp{(\all{X}{R})^{\cup}}_\gamma = \interp{\all{X}{R^{\cup}}}_\gamma$
  \item $\interp{(R_1\to R_2)^{\cup}}_\gamma = \interp{R_1^{\cup} \to R_2^{\cup}}_\gamma$
  \item $\interp{(R_1\cdot R_2)^{\cup}}_\gamma = \interp{R_2^{\cup} \cdot R_1^{\cup}}_\gamma$
  \end{enumerate}
\end{lemma}
\begin{opt}
  \begin{proof}
  \pcase{(1)} assume $t\ \interp{(\all{X}{R})^{\cup}}_\gamma\ t'$, and hence
  $t'\ \interp{(\all{X}{R})}_\gamma\ t$.  For any $r\in\mathcal{R}$, 
  $t'\ \interp{R}_{\gamma[X\mapsto r]}\ t$, hence $t\ \interp{R^{\cup}}_{\gamma[X\mapsto r]}\ t'$.
  From this, $t\ \interp{\all{X}{R^{\cup}}}_\gamma\ t'$ as required.  Conversely,
  assume $t\ \interp{\all{X}{R^{\cup}}}_\gamma\ t'$, and $r\in\mathcal{R}$.  Then $t\ \interp{R^{\cup}}_{\gamma[X\mapsto r]}\ t'$,
  hence $t'\ \interp{R}_{\gamma[X\mapsto r]}\ t$.  From this, $t'\ \interp{\all{X}{R}}_{\gamma}\ t$, hence
  the required $t\ \interp{(\all{X}{R})^{\cup}}_{\gamma}\ t'$.

  \pcase{(2)} 
  Assume $t\ \interp{(R_1\to R_2)^{\cup}}_\gamma\ t'$ and $a\ \interp{R_1^{\cup}}_\gamma\ a'$.
  From these, $t'\ \interp{(R_1\to R_2)}_\gamma\ t$ and $a'\ \interp{R_1}_\gamma\ a$, which yield
  $t'\ a'\ \interp{R_1}_\gamma\ t\ a$.  Thus, $t\ a\ \interp{R_1^{\cup}}_\gamma\ t'\ a'$ as required.
  Conversely, assume $t\ \interp{R_1^{\cup}\to R_2^{\cup}}_\gamma\ t'$ and $a\ \interp{R_1}_\gamma\ a'$.
  From the latter, $a'\ \interp{R_1^{\cup}}_\gamma\ a$, so $t\ a'\ \interp{R_2^{\cup}}_\gamma\ t'\ a$.
  From this, $t'\ a\ \interp{R_2}_\gamma\ t\ a'$, as required.

  \pcase{(3)} assume $t\ \interp{(R_1 \cdot R_2)^{\cup}}_\gamma\ t'$, hence $t'\ \interp{R_1 \cdot R_2}_\gamma\ t$.
  So there exists $t''$ with $t'\ \interp{R_1}_\gamma\ t''$ and $t''\ \interp{R_2}_\gamma\ t$.
  From these, $t''\ \interp{R_1^{\cup}}_\gamma\ t'$ and $t\ \interp{R_2^{\cup}}_\gamma\ t''$; thus,
  $t\ \interp{R_2^{\cup}\cdot R_1^{\cup}}_\gamma\ t'$.  Conversely, assume $t\ \interp{R_2^{\cup} \cdot R_1^{\cup}}_\gamma\ t'$.
  So there exists $t''$ with $t\ \interp{R_2^{\cup}}_\gamma\ t''$ and $t''\ \interp{R_1^{\cup}}_\gamma\ t'$.
  From these, $t''\ \interp{R_2}_\gamma\ t$ and $t'\ \interp{R_1}_\gamma\ t''$.  So
  $t\ \interp{(R_1 \cdot R_2)^{\cup}}_\gamma\ t'$.
  \end{proof}
\end{opt}

\begin{lemma}[Deapplication]
\label{lem:deapp}

\

\begin{enumerate}
\item $t_1\ \interp{t \cdot R}_\gamma\ t_2 = t\ t_1\ \interp{R}_\gamma\ t_2$
\item $t_1\ \interp{R \cdot t^\cup}_\gamma\ t_2 = t_1\ \interp{R}_\gamma\ t\ t_2$
  \end{enumerate}
\end{lemma}
  \begin{proof}
    For the first fact: first, assume $t_1\ \interp{t \cdot R}_\gamma\ t_2$.
    The semantics gives $t'$ such that (1) $t_1\ \interp{t}_\gamma\ t'$ and (2) $t'\ \interp{R}_\gamma\ t_2$.
    But (1) is equivalent to $t\ t_1 =_{\beta\eta} t'$.  Applying $\beta\eta$-Closure~\rref{lem:closedr}, $t\ t_1 \ \interp{R}_\gamma\ t_2$
    as required.  Next, assume $t\ t_1\ \interp{R}_\gamma\ t_2$.  Then there is a $t'$, namely $t\ t_1$, such
    that $t_1\ \interp{t}_\gamma\ t'$ and $t'\ \interp{R}_\gamma\ t_2$.  Hence $t_1\ \interp{t \cdot R}_\gamma\ t_2$
    as required.

    For the second: assuming $t_1\ \interp{R \cdot t^\cup}_\gamma\ t_2$,
    the semantics gives $t'$ such that (1) $t_1\ \interp{R}_\gamma\ t'$ and (2) $t'\ \interp{t^\cup}_\gamma\ t_2$.
    But (2) is equivalent to $t\ t_2 =_{\beta\eta} t'$.  Applying $\beta\eta$-Closure~\rref{lem:closedr}, $t_1 \ \interp{R}_\gamma\ t\ t_2$
    as required.  Next, assume $t_1\ \interp{R}_\gamma\ t\ t_2$.  Then there is a $t'$, namely $t\ t_2$, such
    that $t_1\ \interp{R}_\gamma\ t'$ and $t'\ \interp{t^\cup}_\gamma\ t_2$.  Hence $t_1\ \interp{R \cdot t^\cup}_\gamma\ t_2$
    as required.
  \end{proof}

We make use of a few definitions for terms in Figure~\ref{fig:lamdefs}.

\begin{figure}
  \[
  \begin{array}{lll}
    I & := & \lam{x}{x} \\
    K & := & \lam{x}{\lam{y}{x}} \\
    t \circ t' & := & \lam{x}{t\ (t'\ x)}
  \end{array}
  \]
  \caption{Some standard definitions and notations for terms, used below}
  \label{fig:lamdefs}
\end{figure}

\begin{lemma}[Relational Laws]
\label{lem:comp}
\ 

\begin{enumerate}
\item $\interp{R_1 \cdot (R_2 \cdot R_3)}_\gamma = \interp{(R_1 \cdot R_2) \cdot R_3}_\gamma$
\item $\interp{(R^{\cup})^{\cup}}_\gamma = \interp{R}_\gamma$
\item $\interp{R\cdot I}_\gamma = \interp{I \cdot R}_\gamma = \interp{R}_\gamma$
  \end{enumerate}
\end{lemma}
\begin{opt}
  \begin{proof}
    (1) follows from the semantics of $\cdot$ as relational composition, (2) from the
    semantics of $\cup$ as relational converse, and (3) from Deapplication~\rref{lem:deapp}
    (applying also $\beta\eta$-Closure~\rref{lem:closedr}).
  \end{proof}
\end{opt}

We may observe that Symmetry Properties~\rref{lem:sym} part (3) and
Relational Laws~\rref{lem:comp} validate the complement- and
union-free axioms of the Calculus of Relations (RelTT omits complement
and union)~\cite{givant06}.

\begin{lemma}[Interpretation Over Substitution]
  \label{lem:interpsubst}
  \[
  \interp{[R/X]R'}_\gamma = \interp{R'}_{\gamma[X\mapsto \interp{R}_\gamma]}
  \]
\end{lemma}
\begin{opt}
  \begin{proof} The proof is by induction on $R'$.  Let $\gamma'$ denote $\gamma[X\mapsto \interp{R}_\gamma]$.

    \case{$X$}
    \[
    \interp{[R/X]X}_\gamma = \interp{R}_\gamma = \interp{X}_{\gamma'}
    \]

    \case{$Y$}
    \[
    \interp{[R/X]Y}_\gamma = \gamma(Y) = \interp{Y}_{\gamma'}
    \]

    \case{$R_1\to R_2$}
    \[
    \begin{array}{l}
      \interp{[R/X](R_1\to R_2)}_\gamma = \\
      \interp{[R/X]R_1}_\gamma \to \interp{[R/X]R_2}_\gamma = \\
      \interp{R_1}_{\gamma'} \to \interp{R_2}_{\gamma'} = \\
      \interp{R_1\to R_2}_{\gamma'}
\end{array}
    \]
    
    \case{$\all{X}{R_1}$}
    \[
    \begin{array}{l}
      \interp{[R/X]\all{X}{R_1}}_\gamma = \\
      \bigcap_{r\in\mathcal{R}}\ \interp{[R/X]R_1}_{\gamma[X\mapsto r]} \\
      \bigcap_{r\in\mathcal{R}}\ \interp{R_1}_{\gamma'[X\mapsto r]} \\
      \interp{\all{X}{R_1}}_{\gamma'}
    \end{array}
    \]

    \case{$R_1^\cup$}
    \[
\begin{array}{l}
    \interp{[R/X](R_1^\cup)}_\gamma = \\
    \interp{[R/X]R_1}_\gamma^\cup = \\
    \interp{R_1}_{\gamma'}^\cup = \\
    \interp{R_1^\cup}_{\gamma'}
\end{array}
\]

\case{$R_1\cdot R_2$}
\[
\begin{array}{l}
  \interp{[R/X](R_1\cdot R_2)}_\gamma = \\
    \interp{[R/X]R_1}_\gamma \cdot \interp{[R/X]R_2}_\gamma = \\
    \interp{R_1}_{\gamma'} \cdot \interp{R_2}_{\gamma'} = \\
    \interp{R_1\cdot R_2}_{\gamma'}
\end{array}
\]

\case{$\hat{t}$}
\[
\interp{[R/X]\hat{t}}_\gamma = \interp{\hat{t}}_\gamma = \interp{\hat{t}}_{\gamma'}
\]

  \end{proof}
\end{opt}

\begin{lemma}[Environment Extension]
  \label{lem:envext}
  \
  
  \begin{enumerate}
  \item If $X\not\in\textit{FV}(R)$, then
    \[
    \interp{R}_{\gamma[X\mapsto r]} = \interp{R}_{\gamma} = \interp{[X/Y]R}_{\gamma[X\mapsto\gamma(Y)]}
    \]
\item   If $R$ is closed, then $\interp{R}_{\gamma} = \interp{R}_{\gamma'}$.
  \end{enumerate}
\end{lemma}
\begin{opt}
  \begin{proof}
    The first fact is by an obvious induction on $R$.  The second follows by iterating the first
    one to shrink $\gamma$ to the empty environment, and then build it back up to $\gamma'$ (recall
    that environments map a finite set of type variables).
  \end{proof}
\end{opt}

\section{Basic examples and definitions}
\label{sec:ex}

\begin{lemma}[Identity]
\label{lem:id}
  $I\ \interp{X \to X}_\gamma\ I$
\end{lemma}
  \begin{proof}
    Assume (1) $t\ [\gamma(X)]\ t'$ and show $I\ t\ [\gamma(X)]\ I\ t'$.
    But this follows from (1) by $\beta\eta$-Closure~\rref{lem:closedr}.
  \end{proof}

\begin{definition}
\label{def:inttp}

\ 

\begin{enumerate}
\item  $[t] R := (K\ t) \cdot R$
\item $R [t] := R \cdot (K\ t)^\cup$
\end{enumerate}
\end{definition}

\noindent We can express within the theory the
property of being related to term $t$ by $R$ with the relational types
$[t]R$ and $R[t]$. In particular, this gives us a form of internalized
typing: for example, we may use the type $[I]\all{X}{X \to X}[I]$ to express the property that
$I$ has the expected polymorphic type.
These notations are to be parsed with highest precedence.

\begin{lemma}[Internalized Typing]
\label{lem:reltpin}
\ 

\begin{enumerate}
\item  $t_1\ \interp{[t] R}_\gamma\ t_2 = t\ \interp{R}_\gamma\ t_2$ 
\item $t_1\ \interp{R [t]}_\gamma\ t_2 = t_1\ \interp{R}_\gamma\ t$
\end{enumerate}
\end{lemma}
\begin{opt}
  \begin{proof}
    For (1), use $\beta\eta$-Closure~\rref{lem:closedr}:
    \[
      (t_1\ \interp{[t] R}_\gamma\ t_2) = (K\ t\ t_1\ \interp{R}_\gamma\ t_2) = (t\ \interp{R}_\gamma\ t_2)
    \]

    For (2), use $\beta\eta$-Closure~\rref{lem:closedr} and also Deapplication~\rref{lem:deapp}:
\[
      (t_1\ \interp{R [t]}_\gamma\ t_2) = (t_1\ \interp{R}_\gamma\ K\ t\ t_2) = (t_1\ \interp{R}_\gamma\ t)
    \]
  \end{proof}
\end{opt}

The following operations are reminiscent of conjugation in group theory:
    
\begin{definition}
  $t_1 . R . t_2 := t_1 \cdot R \cdot t_2^\cup$
\end{definition}

\begin{definition}
  $t \dcdot R := t . R . t$
\end{definition}

\begin{lemma}[Conjugation]
  \label{lem:dcdot}
  \
  
  \begin{enumerate}
    \item $t_1\ \interp{t . R . t'}_\gamma\ t_2 = t\ t_1\ \interp{R}_\gamma\ t'\ t_2$.
    \item $t_1\ \interp{t \dcdot R}_\gamma\ t_2 = t\ t_1\ \interp{R}_\gamma\ t\ t_2$.
\end{enumerate}
\end{lemma}
\begin{opt}
  \begin{proof}
    Apply Deapplication~\rref{lem:deapp}.
  \end{proof}
\end{opt}

We may internalize inclusion of relations as a type,
using term promotions:

\begin{definition}
  $R\subseteq R' := (K\ I) \dcdot (R \to R')$
\end{definition}

\begin{lemma}[Subset]
\label{lem:subset}
  $t_1\ \interp{R \subseteq R'}_\gamma\ t_2$ iff $\interp{R}_\gamma \subseteq \interp{R'}_\gamma$.
\end{lemma}
\begin{opt}
  \begin{proof}
Making use of Conjugation~\rref{lem:dcdot}, deduce
\[
\begin{array}{l}
  t_1\ \interp{R\subseteq R'}_\gamma\ t_2 = \\
  K\ I\ t_1\ \interp{R\to R'}_\gamma\ K\ I\ t_2 = \\
  I\ \interp{R\to R'}_\gamma\ I
\end{array}
\]
\noindent  The semantics
(Figure~\ref{fig:relsem}) states that this latter relational typing is true in environment $\gamma$ iff for
all $(x,x')\in \interp{R}_\gamma$, $(I\ x,I\ x')\in \interp{R'}_\gamma$,
which by $\beta\eta$-Closure~\rref{lem:closedr} is equivalent to $(x,x')\in\interp{R'}_\gamma$.
  \end{proof}
\end{opt}
  
Term promotions also enable us to derive implicit
products~\cite{miquel01}. In traditional type theories, implicit
products are used to express quantifications without corresponding
$\lambda$-abstractions in the subject.  One may think of them as
describing specificational (or ``ghost'') inputs to functions. In
RelTT, we express this by stating that the subject has a function type
but erases its input; i.e., it is of the form $K\, t$ for some $t$.

\begin{definition}
  $R \To R' := K \dcdot (R \to R')$
  \end{definition}

Note in the following theorem the essential feature of implicit
products: we conclude by relating (with $R'$) just $t_1$ and
$t_2$, not their applications to $x$ and $x'$ respectively.

\begin{lemma}[Implicit Product]
\label{lem:impprod}
  $t_1\ \interp{R \To R'}_\gamma\ t_2$ iff for all $(x,x')\in\interp{R}$, $t_1\ \interp{R'}_\gamma\ t_2$.
\end{lemma}
\begin{opt}
  \begin{proof}
    \[
    \begin{array}{l}
      t_1\ \interp{R \To R'}_\gamma\ t_2 = \\
      K\ t_1 \ \interp{R \to R'}_\gamma\ K\ t_2
    \end{array}
    \]
    \noindent And the latter holds iff for all $(x,x')\in\interp{R}$, $K\ t_1\ x\ \interp{R'}_\gamma\ K\ t_2\ x'$.
     By $\beta\eta$-Closure~\rref{lem:closedr}, this is equivalent to $t_1\ \interp{R'}_\gamma\ t_2$.
  \end{proof}
\end{opt}

Finally, using internalized inclusion, we may neatly express equality
of relations as a type:

\begin{definition}
  \label{def:releq}
  $R \releq R' := (R \subseteq R') \cdot (R' \subseteq R)$
\end{definition}

\begin{lemma}[Relational Equality]
\label{lem:releq}
  $t_1\ \interp{R \releq R'}_\gamma\ t_2$ iff $\interp{R}_\gamma = \interp{R'}_\gamma$.
\end{lemma}
\begin{opt}
  \begin{proof}
    First, suppose $t_1\ \interp{R \releq R'}_\gamma\ t_2$.  Then by semantics of composition,
    there exists some $t$ such that
    \begin{itemize}
    \item $t_1\ \interp{R\subseteq R'}_\gamma\ t$, and
    \item $t\ \interp{R\subseteq R'}_\gamma\ t_2$.
    \end{itemize}
    \noindent Applying Subset~\rref{lem:subset}, these facts are equivalent to
    \begin{itemize}
    \item $\interp{R}_\gamma \subseteq \interp{R'}_\gamma$, and
    \item $\interp{R'}_\gamma \subseteq \interp{R}_\gamma$.
    \end{itemize}
    \noindent This proves the two relations are equal.

    Next, suppose $\interp{R}_\gamma = \interp{R'}_\gamma$.  Then similarly, applying Subset~\rref{lem:subset},
    we may arbitrarily choose $I$ for
    $t$ to satisfy
    \begin{itemize}
    \item $t_1\ \interp{R\subseteq R'}_\gamma\ t$, and
    \item $t\ \interp{R\subseteq R'}_\gamma\ t_2$.
    \end{itemize}
    \noindent which suffices, again by the semantics of composition.
  \end{proof}
\end{opt}

\begin{lemma}[Substitutivity Of Relational Equality]
\label{lem:substreleq}
  If $t_1\ \interp{R \releq R'}_\gamma\ t_2$, then
  $\interp{[R/X]R''}_\gamma = \interp{[R'/X]R''}_\gamma$
\end{lemma}
\begin{opt}
  \begin{proof}
    The proof is by induction on $R''$, making use of Environment Extension~\rref{lem:envext} as
    we induct on the bodies of universal types (in extended environments).  We omit
    the details, as all cases are obvious thanks to the compositionality of the
    semantics (Figure~\ref{fig:relsem}).
  \end{proof}
\end{opt}

\section{Extensionality principles}
\label{sec:ext}

We prove a few examples of standard type-specific extensionality principles.

\begin{lemma}[$\eta$-Unit]
  \label{lem:etaunit}
  If $t\,\interp{\all{X}{X\to X}}_\gamma\,t'$, then $t\,\interp{\all{X}{X\to X}}_\gamma\,I$.
\end{lemma}
\begin{proof}
  Assume (1) $t\,\interp{\all{X}{X\to X}}_\gamma\,t'$.  Next, assume $r\in\mathcal{R}$ with $y\,[r]\,y'$.
  Instantiate (1) with $\interp{X[y']}_{X\mapsto r}$ (note this is a definable relation) to get
  \[
  t\,\interp{X\to X}_{\gamma[X\mapsto\interp{X[y']}_{X\mapsto r}]}\,t'
  \]
  \noindent Simplifying using Interpretation Over Substitution~\rref{lem:interpsubst} and
  also Environment Extension~\rref{lem:envext}, this gives us
  \[
  t\,\interp{X[y']\to X[y']}_{\gamma[X\mapsto r]}\,t'
  \]
  \noindent We may apply this to $y\,[X[y']]\,y'$ which we have from (1) by Internalized Typing~\rref{lem:reltpin}.
  This application yields
  \[
  t\,y\,\interp{X[y']}_{\gamma[X\mapsto r]}\,t'\,y'
  \]
  \noindent Again applying Internalized Typing~\rref{lem:reltpin}, this gives us $t\,y\,[r]\,y'$, as required.
\end{proof}

\begin{definition}
\label{def:pairs}
  \[
  \begin{array}{lll}
    R \times R' & := & \all{X}{(R \to R' \to X) \to X} \\
    \textit{pair} & := & \lam{x}{\lam{y}{\lam{c}{c\,x\,y}}} \\
    (t,t') & := & \textit{pair}\ t\, t' \\
    t.1 & := & t\ \lam{x}{\lam{y}{x}} \\
    t.2 & := & t\ \lam{x}{\lam{y}{y}}
  \end{array}
  \]
\end{definition}

\begin{lemma}[Surjective Pairing]
  \label{lem:surjpair}
  If $t\,\interp{R\times R'}_\gamma\,t'$, then
  \[
  (t.1,t.2)\,\interp{R \times R'}_\gamma\,t'
  \]
\end{lemma}
  \begin{proof}
    Assume (1) $t\,[R\times R']\,t'$.  Then assume $r\in\mathcal{R}$
    and (2) $c\,\interp{R\to R'\to X}_{\gamma[X\mapsto r]}\,c'$, and show
\begin{equation}
  \textit{pair}\,(t.1)\,(t.2)\,c\,[r]\,t'\,c'
  \tag{A}
  \end{equation}

\noindent Instantiate (1) with $\interp{\lam{x}{x\,c}\cdot X}_{[X\mapsto r]}$
    (note this is a definable relation).  Then (A) follows from 
\begin{equation}
    \textit{pair}\,\interp{R\to R' \to X}_{\gamma[X\mapsto \interp{\lam{x}{x\,c}\cdot X}_{[X\mapsto r]}]}\,c'
\tag{B}
\end{equation}
    \noindent Let us apply
    Environment Extension~\rref{lem:envext} implicitly to simplify environments.
    To prove (B), assume (3) $r_1\,\interp{R}_\gamma\,r_1'$ and (4) $r_2\,\interp{R'}_\gamma\,r_2'$,
    and show
    \[
    \textit{pair}\ r_1\, r_2\interp{\lam{x}{x\,c}\cdot X}_{[X\mapsto r]}\,c'\, r_1'\,r_2'
    \]
    \noindent By Deapplication~\rref{lem:deapp}, this is equivalent to
    \[
    \textit{pair}\ r_1\, r_2\,c\, [r]\,c'\, r_1'\,r_2'
    \]
    \noindent By $\beta\eta$-Closure~\rref{lem:closedr}, this is equivalent to
    \[
    c\,r_1\, r_2 [r]\,c'\, r_1'\,r_2'
    \]
    \noindent This follows from (2), (3), and (4) by the semantics.
  \end{proof}

\section{Symmetric types}
\label{sec:symm}

\begin{definition}
\label{def:symmtp}
Call a type \emph{symmetric} iff it does not use $R \cdot R'$, and
every occurrence of a promotion of a term $t$ either has $t
=_{\beta\eta} I$ or occurs as $t$ in subexpressions of the form
$t \dcdot R$. Use $S$ as a metavariable for symmetric types.
\end{definition}

\begin{definition}
  $\gamma^{\cup}(X) = (\gamma(X))^{\cup}$; i.e., the converse of relation $\gamma(X)$.
\end{definition}

\begin{lemma}
  $(\gamma^{\cup})^{\cup} = \gamma$.
\end{lemma}

\begin{theorem}[Symmetric types]
\label{thm:symmtp}
 $\interp{S}_\gamma = \interp{S^{\cup}}_{\gamma{\cup}}$
\end{theorem}
\begin{opt}
\begin{proof}
  The proof is by induction on $S$.

  \case{$X$} $\interp{X}_\gamma = \gamma(X) = (\gamma^{\cup})^{\cup}(X) = \interp{X^{\cup}}_{\gamma^{\cup}}$.

\case{$S\to S'$} assume $t\ \interp{S\to S'}_\gamma\ t'$. To show $t'\ \interp{S\to S'}_{\gamma^{\cup}}\ t$, assume $a\ \interp{S}_{\gamma^{\cup}}\ a'$.
By the IH, $a'\ \interp{S}_{\gamma}\ a$, so $t\ a'\ \interp{S'}_\gamma\ t'\ a$.  By the IH again, $t'\ a\ \interp{S'}_{\gamma^{\cup}}\ t\ a'$, as
required. Conversely, assume $t\ \interp{S\to S'}_{\gamma^{\cup}}\ t'$, and assume $a\ \interp{S}_\gamma\ a'$.  By the IH,
$a'\ \interp{S}_{\gamma^{\cup}}\ a$, so $t\ a'\ \interp{S'}_{\gamma^{\cup}}\ t'\ a$.  By the IH again, $t'\ a\ \interp{S'}_{\gamma}\ t\ a'$, as required.

\case{$\all{X}{S}$} assume $t\ \interp{\all{X}{S}}_\gamma\ t'$, and $r\in\mathcal{R}$.  So $t\ \interp{S}_{\gamma[X\mapsto r^{\cup}}\ t'$,
and by the IH, $t'\ \interp{S}_{\gamma^{\cup}[X\mapsto r]}\ t$, as required.  Conversely, assume $t\ \interp{\all{X}{S}}_{\gamma^{\cup}}\ t'$,
and $r\in\mathcal{R}$.  So $t\ \interp{S}_{\gamma^{\cup}[X\mapsto r^{\cup}]}\ t'$,
  and by the IH, $t'\ \interp{S}_{\gamma[X\mapsto r]}\ t$, as required.

  \case{$S^{\cup}$} assume $t\ \interp{S^{\cup}}_\gamma\ t'$.  By the IH,
$t\ \interp{S}_{\gamma^{\cup}}\ t'$ as required.  Conversely, assume $t\ \interp{S}_{\gamma^{\cup}}\ t'$.  Then by the IH, $t\ \interp{S^{\cup}}_{\gamma}\ t'$,
as required. 

\case{$\hat{t} \dcdot S$} assume $t\ \interp{\hat{t} \dcdot S}_\gamma\ t'$.  By Conjugation~\rref{lem:dcdot}, this
is equivalent to $\hat{t}\ t\ \interp{S}_\gamma\ \hat{t}\ t'$.  By the IH, $\hat{t}\ t'\ \interp{S}_\gamma\ \hat{t}\ t$,
which is then similarly equivalent to the desired typing.  The converse follows similarly, applying Symmetry Properties~\rref{lem:sym}
and Relational Laws~\rref{lem:comp}.

\case{$\hat{t} =_{\beta\eta} I$} 
\[
\begin{array}{l}
  (t\ \interp{\hat{t}}_\gamma\ t') = (\hat{t}\ t =_{\beta\eta} t') = (t =_{\beta\eta} t') = (t =_{\beta\eta} \hat{t}\ t') = \\  
  (t'\ \interp{\hat{t}}_\gamma\ t)
\end{array}
\]

\end{proof}
\end{opt}

\section{Transitive types}
\label{sec:trans}

\begin{definition}
\label{def:polarity}
Use metavariable $p$ to range over the set  $\{-,+\}$ of \emph{polarities}.  $\bar{p}$ denotes the other polarity
from $p$.
\end{definition}

The following notion extends a similar one due to Krivine~\cite[Section 8.5]{krivine93},
put also to good use in other works like~\cite{pistone19}.

\begin{definition}[$\forall^p$]
Define a property $\forall^p$ of types inductively by the following
clauses.  Type variables $X$ are $\forall^p$.  If $R$ is
$\forall^{\bar{p}}$ and $R'$ is $\forall^{p}$, then $R\to R'$ is
$\forall^p$.  If $R$ is $\forall^+$ then so is $\all{X}{R}$.
If $R$ is $\forall^p$, then so is $R^{\cup}$.  If $t =_{\beta\eta} I$,
then the promotion of $t$ to a type is $\forall^p$.  Note that
$\forall^p$ types are symmetric types (Definition~\ref{def:symmtp}).
We let $P$ range over $\forall^+$ types, and $N$ over $\forall^-$ types.
\end{definition}

Recall the following fact from classical lambda calculus
(e.g., Chapter 7 of~\cite{hs08}).

\begin{lemma}[Zeta]
\label{lem:zeta}
  If $t\ x =_{\beta\eta} t'\ x$ and $x\not\in\textit{FV}(t\,t')$, then $t =_{\beta\eta} t'$.
\end{lemma}
\begin{opt}
  \begin{proof}
    From the assumption, deduce $\lam{x}{t\ x} =_{\beta\eta} \lam{x}{t'\ x}$.
    The sides of this equation are $\eta$-equal to $t$ and $t'$, respectively.
  \end{proof}
\end{opt}

\begin{definition}
  Let $e$ denote the environment where $e(X)$ is the relation $=_{\beta\eta}$, for all type
variables $X$.
\end{definition}

As discussed further in Section~\ref{sec:related}, RelTT by design
does not satisfy Identity Extension (a property proposed originally
by Reynolds~\cite{reynolds83}).  The following is a partial
refinement:

\begin{theorem}[Identity Inclusion]
\label{thm:posneg}

\

\begin{enumerate}
\item $\interp{P}_e \subseteq\ =_{\beta\eta}$.
\item $=_{\beta\eta}\ \subseteq \interp{N}_e$.
  \end{enumerate}
\end{theorem}
\begin{proof}
Proceed by induction on the assumption of $R$ in $\forall^p$.

\case{$X \in \forall^p$} $\interp{X}_e = e(X)$, which is $=_{\beta\eta}$. 

\case{$R\to R'\in\forall^+$}
assume (1) $t\ \interp{R \to R'}_e\ t'$.  By Zeta~\rref{lem:zeta}, it suffices to prove $t\ x =_{\beta\eta} t'\ x$.
Since $R\in\forall^-$, the IH applies to $x =_{\beta\eta} x$ to yield $x\ \interp{R}_e\ x$.  Combining this
with (1) gives $t\ x\ \interp{R'}_e\ t'\ x$.  Then by the IH, $t\ x =_{\beta\eta} t'\ x$, as required.

\case{$R\to R'\in\forall^-$} assume (1) $t =_{\beta\eta} t'$ and (2) $a\ \interp{R}_e\ a'$,
and show $t\ a\ \interp{R'}_e\ t'\ a'$.  Since $R \in\forall^+$, the IH applies to (2) yielding
$a =_{\beta\eta} a'$.  Combining this with (1) gives $t\ a =_{\beta\eta} t'\ a'$,
from which the IH yields $t\ a\ \interp{R'}_e\ t'\ a'$.  

\case{$\all{X}{R}\in\forall^+$} assume
(1) $t\ \interp{\all{X}{R}}_e\ t'$, and show $t =_{\beta\eta} t'$.  From (1),
we have $t\ \interp{R}_{e[X\mapsto\ =_{\beta\eta}]}\ t'$.  By the IH, this yields $t =_{\beta\eta} t'$, as
required. 

\case{$R^{\cup}\in\forall^+$} assume $t\ \interp{R^{\cup}}_e\ t'$, which implies
$t'\ \interp{R}_e\ t$.  By the IH, $t' =_{\beta\eta} t$, hence $t =_{\beta\eta} t'$ as required.

\case{$R^{\cup}\in\forall^-$} assume $t =_{\beta\eta} t'$, hence $t' =_{\beta\eta} t$.
By the IH, $t'\ \interp{R}_e\ t$, which equals the required $t\ \interp{R^{\cup}}_e\ t'$.

\case{$\hat{t}\in\forall^p$} $\interp{\hat{t}}_e$ is then just $=_{\beta\eta}$.

\end{proof}

Using the terminology of~\cite{wadler07}, Identity Inclusion~\rref{thm:posneg} identifies
$\forall^+$ types as \emph{extensive} (they are included in the
equality relation), and $\forall^-$ types as \emph{parametric} (the
equality relation is included in them).

\begin{lemma}[Transitivity For $\forall^+$-Types]
\label{lem:transptp}
  $I\ \interp{P \cdot P\to P}_e\ I$.
\end{lemma}

  \begin{proof}
    Assume (1) $x\ \interp{R}_e\ y$ and (2) $y\ \interp{R}_e\ z$, and show $x\ \interp{R}_e\ z$.
    By Identity Inclusion~\rref{thm:posneg}, (1) implies $x\ =_{\beta\eta}\ y$.  From this and (2), $\beta\eta$-Closure~\rref{lem:closedr}
    yields the desired conclusion.
  \end{proof}

\begin{corollary}[$\forall^+$ Per]
  \label{cor:forallper}
  If $R$ is $\forall^+$ and closed, then $\interp{R}_\gamma$ is a partial
  equivalence relation (i.e., symmetric and transitive; abbreviated per).
\end{corollary}
\begin{proof}
    Since $R$ is closed, $\interp{R}_\gamma = \interp{R}_e$ by Environment Extension~\rref{lem:envext}.
    Transitivity for $\forall^+$-Types~\rref{lem:transptp} then implies transitivity.  Symmetry follows from
    Symmetric Types~\rref{thm:symmtp},
    since $\forall^+$ types are symmetric types (Definition~\ref{def:symmtp}).
\end{proof}

\begin{definition}[simple transitive types]
  Simple transitive types $T$ are defined by the following grammar:
\[
T\ ::= \ P\ |\ P \to T\ |\ N \to T\ |\ t\dcdot T
\]
\end{definition}

\begin{lemma}[transitivity for simple transitive types]
  $I\ \interp{T \cdot T \to T}_e I$
\end{lemma}
  \begin{proof}
    The proof is by induction on $T$, in each case assuming (1) $x\ \interp{T}_\gamma\ y$ and (2) $y\ \interp{T}_\gamma\ z$.

    \case{$P$} Transitivity for $\forall^+$-Types~\rref{lem:transptp}.

    \case{$P \to T$} assume (3) $a\ \interp{P}_e\ a'$.  By Symmetric Types~\rref{thm:symmtp}, $a'\ \interp{P}_e\ a$ (as $e^{\cup} = e$).  By
    Transitivity for $\forall^+$-Types~\rref{lem:transptp}, this can be combined
    with (3) to obtain $a\ \interp{P}_e\ a$.  Using this with (1), $x\ a\ \interp{T}_e\ y\ a$.  Using (3) with (2), $y\ a\ \interp{T}_e\ z\ a'$.
    By the induction hypothesis, $x\ a\ \interp{T}_e\ z\ a'$ as required.

    \case{$N \to T$} assume (3) $a\ \interp{N}_e\ a'$.  By Identity Inclusion~\rref{thm:posneg}, since $N$ is $\forall^-$, $a\ \interp{N}_e\ a$ (since
    $a =_{\beta\eta} a$).  Using this with (1), $x\ a\ \interp{T}_e\ y\ a$.  Then as in the previous case, we obtain $y\ a\ \interp{T}_e\ z\ a'$
    using (3) with (2), and the required $x\ a\ \interp{T}_e\ z\ a'$ by the induction hypothesis.

    \case{$\hat{t}\dcdot T$} by Conjugation~\rref{lem:dcdot}, it suffices to show $\hat{t}\ x\ \interp{T}_\gamma\ \hat{t}\ z$. This
    follows by the IH from assumptions (1) and (2), since these are equivalent to $\hat{t}\ x\ \interp{T}_\gamma\ \hat{t}\ y$
    and $\hat{t}\ y\ \interp{T}_\gamma\ \hat{t}\ z$ by Conjugation~\rref{lem:dcdot}.
  \end{proof}

\section{A relational proof system}
\label{sec:relpf}

Figure~\ref{fig:relpf} presents a proof system, RelPf, for judgments of the
form $\Gamma \vdash t\ [R]\ t'$.  (Here, the square brackets are part of the
syntax for the judgment; in our meta-language, we are using them for application of a mathematical relation.)
RelPf Soundness~\rref{thm:snd} below
shows that this system is sound with respect to the semantics of
Figure~\ref{fig:relsem} (extended for contexts).
In Section~\ref{sec:relty}, we will develop a type theory based
on RelPf, but introduce the proof system here
because the fragment for System F types will be useful in 
Section~\ref{sec:indtp} on inductive types.   A few details:

\begin{itemize}
\item typing contexts $\Gamma$ are described by the grammar
  \[
  \Gamma\ ::= \ \cdot\ |\ \Gamma, t\,[R]\,t'
  \]
  \noindent We may elide $\cdot$ in examples.

\item There is an introduction and elimination rule for each connective.
\item The introduction rule for term promotions is the axiom $\Gamma\vdash t\,[t']\,t'\,t$.
  This states that $t$ is related to $t'\,t$ by the relation (i.e., term promotion) $t'$.
\item The rule allowing to change the sides of the relational typing to $\beta\eta$-equal terms is called \emph{conversion}.
  While $\beta\eta$-equality is undecidable in general, we may view the side conditions on conversion as
  license for an implementation to check reductions to as deep a finite depth as desired.  So we view
  reduction as being implicitly bounded in applications of this rule, making type-checking decidable.
  We do not formalize bounded reduction.
\end{itemize}

  \begin{figure}[!t]
    \[
    \infer{\Gamma \vdash x\,[R]\,y'}
          {\infer{\Gamma\vdash \textit{tt}\,x\,y\,[R]\,\textit{ff}\,x'\,y'}
            {\infer{\Gamma\vdash \textit{tt}\,x\,[R\to R]\,\textit{ff}\,x'\,y'}
              {\infer{\Gamma\vdash \textit{tt}\,[R\to R\to R]\, \textit{ff}}
                     {\infer{\Gamma\vdash \textit{tt}\,[\textit{Bool}]\,\textit{ff}}{\ }}
                & \infer{\Gamma\vdash x\,[R]\,x'}{\ }}
              & \infer{\Gamma\vdash y\,[R]\,y'}{\ }}}
    \]
    \caption{Derivation of True Different From False~\rref{lem:tf}.  The final inference is by the conversion rule,
    noting $\textit{tt}\,x\,y =_{\beta\eta} x$ and $\textit{ff}\,x'\,y' =_{\beta\eta} y'$}
    \label{fig:tf}
\end{figure}

\begin{figure}[!t]
  \[
  \begin{array}{ll}
    \infer{\Gamma\vdash x : T}{x:T\in\Gamma}
    &
    \infer{\Gamma\vdash \lam{x}{t} : T \to T'}{\Gamma,x:T\vdash t : T'}
    \\
    \\
    \infer{\Gamma\vdash t\ t' : T}{\Gamma\vdash t : T'\to T & \Gamma\vdash t':T'}
    &
    \infer{\Gamma \vdash t : \all{X}{T}}{\Gamma \vdash t : T & X \not \in \textit{FV}(\Gamma)}
    \\
    \\
    \infer{\Gamma \vdash t : [T'/X]T}{\Gamma \vdash t : \all{X}{T}}
    &
    \ 
  \end{array}
  \]
  \caption{Typing rules for Curry-style System F}
  \label{fig:lamtwo}
\end{figure}

\begin{figure*}[!t]
  \[
  \begin{array}{lllll}
    \infer{\Gamma\vdash t\,[R]\,t'}{t\,[R]\,t' \in \Gamma}
    &\ &
    \infer{\Gamma\vdash  \lam{x}{t}\,[R \to R']\,\lam{x'}{t'}}
          {\Gamma, x\,[R]\,x' \vdash t\,[R']\,t' & (*)}
  & \ &\infer{\Gamma\vdash t\,t_1\,[R']\,t'\,t_2}
          {\Gamma \vdash t\,[R \to R']\,t'
            & \Gamma \vdash t_1\,[R]\,t_2} 

\\\\

          \infer{\Gamma\vdash t\,[[R/X]R']\,t'}
                {\Gamma \vdash t\,[\all{X}{R'}]\,t'}
&\ &
          \infer{\Gamma\vdash t\,[\all{X}{R}]\,t'}
                {\Gamma \vdash t\,[R]\,t' & X\not\in\textit{FV}(\Gamma)}

&\ &
         \infer{\Gamma \vdash t_1'\,[R]\,t_2'}
               {\Gamma \vdash t_1\,[R]\,t_2 & t_1 =_{\beta\eta} t_1' & t_2 =_{\beta\eta} t_2'}

\\ \\
         \infer{\Gamma \vdash t\,[R^\cup]\,t'}
                {\Gamma \vdash t'\,[R]\,t}           
&\ & 
         \infer{\Gamma \vdash t'\,[R]\,t}           
               {\Gamma \vdash t\,[R^\cup]\,t'}

\ \ \ \ 

          \infer{\Gamma \vdash t\,[t']\,t'\ t}
                {\ }
                
  &\ &
         \infer{\Gamma \vdash [t''\,t/x]t_1\,[R]\,[t''\,t/x]t_2}
               {\Gamma \vdash t\,[t'']\,t' & \Gamma \vdash [t'/x]t_1\,[R]\,[t'/x]t_2}
\\ \\
      \multicolumn{3}{l}{\infer{\Gamma \vdash t_1\,[R'']\,t_2}
              {\Gamma \vdash t\,[R\cdot R']\,t' & \Gamma , t\,[R]\,x, x\,[R']\,t'\vdash t_1\,[R'']\,t_2 & (**)}}
   &\ &     \infer{\Gamma \vdash t\,[R\cdot R']\,t'}
               {\Gamma \vdash t\,[R]\,t'' & \Gamma \vdash t''\,[R']\, t'}
  
  \end{array}
  \]
  \hspace{1.8cm}Side condition (*) is $x\not\in\textit{FV}(\Gamma, R, R')$.

  \hspace{1.8cm}Side condition (**) is $x\not\in\textit{FV}(\Gamma, t_1, t_2, t, t', R, R', R'')$.
  \caption{Proof system for relational typing. }
  \label{fig:relpf}
\end{figure*}

  Here is an example in RelPf, deriving a form of inconsistency from an
  assumption that different constructors of an inductive type are equal.
  It states that if \emph{tt} and \emph{ff} are equal as booleans,
  then any relation $R$ is trivial in the sense that $R = \textit{dom}(R)\times\textit{ran}(R)$.

\begin{definition}
  \[
  \begin{array}{lll}
    \textit{Bool} & := & \all{X}{X\to X \to X} \\
    \textit{tt} & := & \lam{x}{\lam{y}{x}} \\
    \textit{ff} & := & \lam{x}{\lam{y}{y}}
  \end{array}
  \]
\end{definition}

  \begin{lemma}[True Different From False]
    \label{lem:tf}

    For any type $R$, let $\Gamma$ be a context with the following assumptions:
\begin{enumerate}
  \item $\textit{tt}\,[\textit{Bool}]\, \textit{ff}$
  \item $x\,[R]\,x'$
  \item $y\,[R]\,y'$
    \end{enumerate}

\noindent Then $\Gamma \vdash x\,[R]\,y'$.
  \end{lemma}
  \begin{proof}
    A derivation is in Figure~\ref{fig:tf}.
    \end{proof}

  Turning now to meta-theory: let $\sigma$ range over term
  substitutions (finite functions from term variables to terms).
  Denote capture-avoiding application of a substitution $\sigma$ to a
  term $t$ as $\sigma\ t$.  Apply substitutions $\sigma$ to types $R$
  by applying them to all terms contained in $R$.  Now we will define
  an interpretation of contexts $\Gamma$ as sets of substitutions
  satisfying the contexts constraints.

\begin{definition}
$\interp{\Gamma}_\gamma$ is defined by recursion on $\Gamma$:
  \[
  \begin{array}{lll}
    \sigma\in \interp{\Gamma,t\,[R]\,t'}_\gamma & = & \sigma\in\interp{\Gamma}_\gamma\ \wedge\ \sigma t\,\interp{\sigma R}_\gamma\,\sigma t'\\
    \sigma\in \interp{\cdot}_\gamma & = & \textit{True}
  \end{array}
  \]
  \end{definition}

\begin{theorem}[RelPf Soundness]
  \label{thm:snd}
  Suppose $\gamma$ is defined on all free type variables of $\Gamma$ and $R$.
  If $\Gamma \vdash t\,[R]\,t'$, and $\sigma\in\interp{\Gamma}_\gamma$, then
  $\sigma\ t\,\interp{\sigma\, R}_\gamma\,\sigma\ t'$.
\end{theorem}
\begin{opt}
  \begin{proof}
    The proof is by induction on the RelPf derivation.  In each case we assume arbitrary $\sigma\in\interp{\Gamma}_\gamma$.

    \caseb
    \[
    \infer{\Gamma\vdash t\,[R]\,t'}{t\,[R]\,t' \in \Gamma}
    \]
    \noindent From $t\,[R]\,t'\in\Gamma$ we obtain the desired $\sigma\, t\,\interp{\sigma\, R}_\gamma\sigma\,t'$ from
    the semantics of contexts.

    \caseb
    \[
    \infer{\Gamma\vdash  \lam{x}{t}\,[R \to R']\,\lam{x'}{t'}}
          {\Gamma, x\,[R]\,x' \vdash t\,[R']\,t' & (*)}
    \]
    \noindent Assume arbitrary $t_1$ and $t_2$ with (1) $t_1\,\interp{\sigma\, R}_\gamma\,t_2$.
    Let $\sigma'$ denote $\sigma[x\mapsto t_1,x'\mapsto t_2]$.  By the IH,
    \[
    \sigma'\,t\interp{\sigma'\, R'}_\gamma\,\sigma'\,t'
    \]
    \noindent By side condition (*), $\sigma'\, R' = \sigma\, R$.  Applying then $\beta\eta$-Closure~\rref{lem:closedr}, we have
    \[
    (\sigma\,\lam{x}{t})\, t_1\interp{\sigma\, R'}_\gamma\,(\sigma\,\lam{x'}{t'})\, t_2
    \]
    \noindent By the semantics of arrow types, the fact that this holds for all $t_1$ and $t_2$ satisfying (1)
    implies the desired $\sigma\, \lam{x}{t}\,\interp{\sigma\, (R \to R')}_\gamma\,\sigma\, \lam{x'}{t'}$.

    \caseb
    \[
    \infer{\Gamma\vdash t\,t_1\,[R']\,t'\,t_2}
          {\Gamma \vdash t\,[R \to R']\,t'
            & \Gamma \vdash t_1\,[R]\,t_2}
     \]
     \noindent By the IH, $\sigma\,t\,\interp{R\to R'}_\gamma\,\sigma\,t'$ and
     $\sigma\,t_1\,\interp{R}_\gamma\,\sigma\,t_2$.  The semantics of arrow types
     then gives the desired $\sigma\,(t\,t_1)\,\interp{R'}_\gamma\,\sigma\,(t'\,t_2)$.

     \caseb
     \[
          \infer{\Gamma\vdash t\,[[R/X]R']\,t'}
                {\Gamma \vdash t\,[\all{X}{R'}]\,t'}
     \]
     \noindent By the IH, we have (1) $\sigma\,t\,\interp{\sigma\,\all{X}{R'}}_\gamma\,\sigma\,t'$.
     By the condition on $\gamma$, $\interp{R}_\gamma$ is defined, and we use it to instantiate (1).
     This gives
     \[
     \sigma\,t\,\interp{\sigma\,R'}_{\gamma[X\mapsto\interp{R}_\gamma]}\,\sigma\,t'
     \]
     \noindent By Interpretation Over Substitution~\rref{lem:interpsubst}, this implies
     the desired 
     \[
     \sigma\,t\,\interp{\sigma\,[R/X]R'}_{\gamma}\,\sigma\,t'
     \]

     \caseb
     \[
         \infer{\Gamma\vdash t\,[\all{X}{R}]\,t'}
                {\Gamma \vdash t\,[R]\,t' & X\not\in\textit{FV}(\Gamma)}
     \]
     \noindent Assume arbitrary $r\in\mathcal{R}$.  Then by the IH, $\sigma\,t\,\interp{\sigma\,R}_{\gamma[X\mapsto r]}\,\sigma'\,t'$.
     The desired $\sigma\,t\,\interp{\sigma\,\all{X}{R}}_\gamma\,\sigma'\,t'$ then follows by the semantics of universal quantification.

     \caseb
     \[
         \infer{\Gamma \vdash t_1'\,[R]\,t_2'}
               {\Gamma \vdash t_1\,[R]\,t_2 & t_1 =_{\beta\eta} t_1' & t_2 =_{\beta\eta} t_2'}
     \]
     \noindent This case follows easily by the IH and $\beta\eta$-Closure~\rref{lem:closedr}.

     \caseb
     \[
         \infer{\Gamma \vdash t\,[R^\cup]\,t'}
                {\Gamma \vdash t'\,[R]\,t}           
     \]
     \noindent By the IH, $\sigma\,t'\interp{R}_\gamma\sigma\,t$.  By the semantics
     of converse, this implies the required $\sigma\,t\interp{R^\cup}_\gamma\sigma\,t'$.

     \caseb
     \[
     \infer{\Gamma \vdash t'\,[R]\,t}
           {\Gamma \vdash t\,[R^\cup]\,t'}
     \]
     \noindent By the IH, $\sigma\,t\interp{R^\cup}_\gamma\sigma\,t'$.  By the semantics
     of converse, this implies the required $\sigma\,t'\interp{R}_\gamma\sigma\,t$.

     \caseb
     \[
          \infer{\Gamma \vdash t\,[t']\,t'\ t}
                {\ }
     \]
     \noindent The desired conclusion is equivalent to $\sigma\,(t'\,t) =_{\beta\eta} \sigma\,(t'\ t)$, which holds.

     \caseb
     \[
         \infer{\Gamma \vdash [t''\,t/x]t_1\,[R]\,[t''\,t/x]t_2}
               {\Gamma \vdash t\,[t'']\,t' & \Gamma \vdash [t'/x]t_1\,[R]\,[t'/x]t_2}
     \]
     \noindent By the IH, we have
     \begin{itemize}
     \item $\sigma\,(t''\,t) =_{\beta\eta} \sigma\,t'$
     \item $\sigma\, [t'/x]t_1\,\interp{\sigma\, R}_\gamma\,\sigma\,[t'/x]t_2$
     \end{itemize}
     \noindent Using basic properties of $\beta\eta$-equivalence and substitution, these
     facts imply the desired
     \[
     \sigma\, [t''\,t/x]t_1\,\interp{\sigma\, R}_\gamma\,\sigma\,[t''\,t/x]t_2
     \]
     
     \caseb
     \[
        \infer{\Gamma \vdash t_1\,[R'']\,t_2}
              {\Gamma \vdash t\,[R\cdot R']\,t' & \Gamma , t\,[R]\,x, x\,[R']\,t'\vdash t_1\,[R'']\,t_2 & (**)}
     \]
     \noindent By the IH and semantics for composition we have that there exists $t''$ such that
     \begin{itemize}
     \item[(1)] $\sigma\,t\,\interp{\sigma\, R}_\gamma\,t''$
     \item[(2)] $t''\,\interp{\sigma\, R'}_\gamma\,\sigma\,t'$
     \end{itemize}
     \noindent Let $\sigma'$ denote $\sigma[x\mapsto t'']$.  Using (1) and (2), we may prove that $\sigma'$ is
     in the interpretation of the context in the right premise of the inference.  Side condition (**) is used
     to deduce that $\sigma'$ satisfies the two constraints added to $\Gamma$ in that context, from (1) and (2) (where only $\sigma$ appears).
     Then by the IH and (**), we have the required
     \[
     \sigma\,t_1\,\interp{\sigma\,R''}_\gamma\,\sigma\,t_2
     \]

     \caseb
     \[
     \infer{\Gamma \vdash t\,[R\cdot R']\,t'}
               {\Gamma \vdash t\,[R]\,t'' & \Gamma \vdash t''\,[R']\, t'}
     \]
     \noindent By the IH, we have
     \begin{itemize}
     \item $\sigma\,t\,\interp{\sigma\,R}_\gamma\,\sigma\,t''$
     \item $\sigma\,t''\,\interp{\sigma\,R'}_\gamma\,\sigma\,t'$
     \end{itemize}
     \noindent These imply the desired $\sigma\,t\,\interp{\sigma\,(R\cdot R')}_\gamma\,\sigma\,t'$ by the semantics of composition.

  \end{proof}
\end{opt}

\section{Embedding System F}

Similar to the Abstraction Theorem of Reynolds~\cite{reynolds83}, we
may prove that each term typable in System F is related to itself by
the relational interpretation of its type.  Figure~\ref{fig:lamtwo}
recalls the typing rules of Curry-style System F (also known as
$\lambda 2$-\textit{Curry}~\cite{barendregt93}).  We consider the set
of types of System F a subset of the set of relational types
(Figure~\ref{fig:reltypes}).  We first show that typing derivations in
System F can be translated to RelTT in the obvious way.  Then we may
appeal to RelTT Soundness~\rref{thm:snd}.  

  \begin{definition}
    Partition the set of variables by an injection $\dot{-}$.
    Assume $t$ does not contain any variables of the form $\dot{x}$ with $x\in\textit{FV}(t)$.
    Then let $\dot{t}$ be the term where every variable $x$ (free or bound) is renamed to $\dot{x}$.
  \end{definition}

  \begin{definition}
    Define $\ctxtup{-}$ recursively on typing contexts $\Gamma$ of System F by:
    \[
    \begin{array}{lll}
      \ctxtup{\cdot} & = & \cdot \\
      \ctxtup{\Gamma,x:T} & = & \ctxtup{\Gamma}, x\,[T]\,\dot{x}
    \end{array}
    \]
  \end{definition}
  
\begin{theorem}[Soundness Of System F]
  \label{thm:sndf}
  If $\Gamma \vdash t : T$ (in System F), then $\ctxtup{\Gamma} \vdash t\ [T]\ \dot{t}$ (in RelPf), assuming $\dot{t}$ is defined.
\end{theorem}
\begin{opt}
\begin{proof} The proof is by induction on the typing derivation in System F.

  \caseb
  \[
    \infer{\Gamma\vdash x : T}{x:T\in\Gamma}
  \]
  \noindent From $x:T\in\Gamma$ we derive $x\,[T]\,\dot{x}\in\ctxtup{\Gamma}$, and conclude using the assumption
  rule of RelPf.
  
  \caseb
  \[
    \infer{\Gamma\vdash \lam{x}{t} : T \to T'}{\Gamma,x:T\vdash t : T'}
    \]
    \noindent By the IH, we have
    \[
    \ctxtup{\Gamma},x\,[T]\,\dot{x}\vdash t\,[T']\,\dot{t}
    \]
    \noindent From this, use arrow introduction (of RelPf) to derive the desired
    \[
    \ctxtup{\Gamma}\vdash \lam{x}{t}\,[T\to T']\,\lam{\dot{x}}{\dot{t}}
    \]

  \caseb
  \[
    \infer{\Gamma\vdash t\ t' : T}{\Gamma\vdash t : T'\to T & \Gamma\vdash t':T'}
  \]
  \noindent By the IH we have
  \[
  \begin{array}{l}
    \ctxtup{\Gamma}\vdash t\,[T'\to T]\,\dot{t} \\
    \ctxtup{\Gamma}\vdash t'\,[T']\,\dot{t'}
  \end{array}
  \]
  \noindent Use arrow elimination (of RelPf) to deduce the desired
  \[
  \ctxtup{\Gamma}\vdash t\,t'\,[T]\,\dot{t\, t'}
  \]

  \caseb
  \[
    \infer{\Gamma \vdash t : \all{X}{T}}{\Gamma \vdash t : T & X \not \in \textit{FV}(\Gamma)}
  \]
  \noindent By the IH, we have $\ctxtup{\Gamma}\vdash t\,[T]\,\dot{t}$.  Apply forall introduction (of RelPf)
  to conclude the desired $\ctxtup{\Gamma}\vdash t\,[\all{X}{T}]\,\dot{t}$

  \caseb
  \[
    \infer{\Gamma \vdash t : [T'/X]T}{\Gamma \vdash t : \all{X}{T}}
    \]
    \noindent By the IH, we have $\ctxtup{\Gamma}\vdash t\,[\all{X}{T}]\,\dot{t}$.  Apply forall elimination (of RelPf)
    to conclude the desired $\ctxtup{\Gamma}\vdash t\,[[T'/X]T]\,\dot{t}$.

\end{proof}
\end{opt}
  
\begin{corollary}[Soundness Of System F For Closed Terms]
  \label{cor:sndfc}
  If $\cdot \vdash t : T$ (in System F), then $t\ \interp{T}_\gamma\ t$.
\end{corollary}
\begin{opt}
  \begin{proof}
    Use Soundness of System F~\rref{thm:sndf} (noting that $t =_\alpha \dot{t}$ since $t$ closed),
    and then RelTT Soundness~\rref{thm:snd}.
  \end{proof}
\end{opt}

\noindent Below we will also need this basic syntactic property:

\begin{proposition}[Weakening for System F]
  \label{prop:weakf}
  If $\Gamma_1,\Gamma_2\vdash t : T$, then $\Gamma_1,x:R,\Gamma_2\vdash t: T$ where $x$ is not
  declared in $\Gamma_1,\Gamma_2$.
  \end{proposition}

\section{Inductive types}
\label{sec:indtp}

Following a relational, and functorial, generalization
of~\cite{wadler07}, this section shows how to derive a relational form
of induction within RelTT.  For this section, except as noted
in Section~\ref{sec:mono}, let $R$ be
a type of System F, possibly containing specified variable $X$ free.
Under the usual requirement of positivity, we prove equal the
following two relational types, where in the second one, we make use
of our notation for internalized typing (Definition~\ref{def:inttp}):
\begin{definition}
\ 

  \begin{itemize}
  \item $\dparam := \all{X}{(R \to X) \to X}$
  \item $\dind := \all{X}{([\textit{in}_{X,R}]\,(R \to X)\,[\textit{in}_{X,R}]) \To X}$
  \end{itemize}
\end{definition}
\noindent $\textit{in}_{X,R}$ represents the constructors
of the inductive datatype in a standard way, and is defined below (Definition~\ref{def:inf}).

\subsection{Variable Polarity and Monotonicity}
\label{sec:mono}

The first step to proving equality of $\dparam$ and $\dind$ is to
extend the usual notion of a type variable's occurring free only
positively or only negatively, to relational types (recall
Definition~\ref{def:polarity} for polarities $p$).  For inductive
types, our results hold only for $\forall^+$ types of System F.  For positive-recursive types,
however (Section~\ref{sec:rectp}), our derivation works for any relational
type $R$.  So we begin by defining when a variable occurs only with
polarity $p$ ($X\in^p R$) generally for any relational type $R$:

\begin{definition}

  Define  $X \in^p R$ inductively by the clauses:

  \begin{itemize}
  \item $X \in^+ X$
  \item $X \in^p Y$
  \item $X \in^p (R \to R')$ iff $X \in^{\bar{p}} R$ and $X \in^p R'$
  \item $X \in^p \all{Y}{R}$ iff $X \in^p R$
  \item $X \in^p (R \cdot R')$ iff $X \in^p R$ and $X \in^p R'$
  \item $X \in^p (R^\cup)$ iff $X \in^p R$
  \item $X \in^p t$
  \end{itemize}
\end{definition}

\noindent (The intention is that $X \in^+ R$ means $X$ occurs only positively in $R$, and
$X \in^- R$ only negatively.)  The following form of monotonicity then holds for
any relational type.  The statement of the lemma using a polarity meta-variable $p$
consolidates many dual cases in the proof (cf.~\cite{eades16}).

\begin{lemma}[Monotonicity]
\label{lem:mono}
  Suppose $r_+$ and $r_-$ are in $\mathcal{R}$, with $r_+ \subseteq r_-$.  If $X \in^p R$, then
  $\interp{R}_{\gamma[X\mapsto r_p]} \subseteq \interp{R}_{\gamma[X\mapsto r_{\bar{p}}]}$.
\end{lemma}
  \begin{proof} The proof is by induction on $X\in^p R$, assuming (1) $r_+\subseteq r_-$ and
    (2) $t_1\ \interp{R}_{\gamma[X\mapsto r_p]}\ t_2$.
    
    \case{$X \in^+ X$} by (1).

    \case{$X \in^p Y$} by (2), as $\interp{Y}_{\gamma[X\mapsto r_p]} = \interp{Y}_{\gamma} = \interp{Y}_{\gamma[X\mapsto r_{\bar{p}}]}$.

    \case{$X \in^p (R_1 \to R_2)$} 
    assume (3) $t_a\ \interp{R_1}_{\gamma[X\mapsto r_{\bar{p}}]}\ t_b$.  From this, the IH for $R_1$ gives $t_a\ \interp{R_1}_{\gamma[X\mapsto r_p]}\ t_b$
    (instantiating the quantified polarity in the IH with $\bar{p}$).
    Combine this with (2) to
    obtain $t_1\ t_a\ \interp{R_2}_{\gamma[X\mapsto r_p]}\ t_2\ t_b$.
    From this, the IH for $R_2$ gives $t_1\ t_a\ \interp{R_2}_{\gamma[X\mapsto r_{\bar{p}}]}\ t_2\ t_b$, as required.  

    \case{$X \in^p \all{Y}{R'}$} assume $r\in\mathcal{R}$, and instantiate (2) with $r$.  Then apply
    the IH to obtain the required $t_1\ \interp{R'}_{\gamma[X\mapsto r_{\bar{p}},Y\mapsto r]}\ t_2$.

    \case{$X \in^p (R_1 \cdot R_2)$} (2) implies that there exists $t$ such that $t_1\ \interp{R_1}_{\gamma[X\mapsto r_P]}\ t$
    and $t\ \interp{R_2}_{\gamma[X\mapsto r_p]}\ t_2$.  Applying the IH, we obtain $t_1\ \interp{R_1}_{\gamma[X\mapsto r_{\bar{P}}]}\ t$
    and $t\ \interp{R_2}_{\gamma[X\mapsto r_{\bar{p}}]}\ t_2$, which suffices.

    \case{$X \in^p (R_a^\cup)$} (2) implies $t_2\ \interp{R_a}_{\gamma[X\mapsto r_p]}\ t_1$.  From this, the IH gives
    $t_2\ \interp{R_a}_{\gamma[X\mapsto r_{\bar{p}}]}\ t_1$, which suffices.

    \case{$X \in^p t$} by (2), as $\interp{t}_{\gamma[X\mapsto r_p]} = \interp{t}_{\gamma} = \interp{t}_{\gamma[X\mapsto r_{\bar{p}}]}$.

\end{proof}

\subsection{Fmap, Fold, and In}
\label{sec:in}

Following a standard approach to derivation of inductive types (cf.~\cite{wadler90}), we
will define operations $\textit{fmap}_{X,R}$, \textit{fold}, and finally $\textit{in}_{X,R}$,
and prove relational typings about them.  Because we will be considering terms related to themselves,
it is convenient to introduce notation $t :: r$:

\begin{definition}
\label{def:dblcol}
  $(t :: r) := t\ [r]\ t$
\end{definition}

\begin{definition}
  Define a term $\textit{fmap}_{X,R}$ by recursion on types $R$ of System F (also, recall Figure~\ref{fig:lamdefs}):
  \[
  \begin{array}{lll}
    \textit{fmap}_{X,X} & = & I \\
    \textit{fmap}_{X,Y} & = & K\ I \\
    \textit{fmap}_{X,R \to R'} & = & \lam{f}{\lam{a}{\textit{fmap}_{X,R'}\ f\ \circ a \circ \textit{fmap}_{X,R}\ f}} \\    
    \textit{fmap}_{X,\all{Y}{R}} & = & \lam{f}{\textit{fmap}_{X,R}\ f}
  \end{array}
  \]
\end{definition}

\noindent Note that as we treat expressions up to $\alpha$-equivalence, we do not need a case for $\textit{fmap}_{X,\all{X}{R}}$,
as this will be handled as $\textit{fmap}_{X,\all{Y}{[Y/X]R}}$.

\begin{lemma}[Fmap (System F)]
\label{lem:fmapf}
Suppose $X_+$ and $X_-$ are type variables.
Suppose $X_p\not\in\textit{FV}(R)$, for all $p$.  
If $X \in^p R$,
then in System F we have
\[
\cdot \vdash \textit{fmap}_{X,R} : (X_+ \to X_-) \to [X_p/X]R \to [X_{\bar{p}}/X]R
\]
\end{lemma}
  \begin{proof} The proof is by induction on $X\in^p R$, implicitly applying Weakening for System F~\rref{prop:weakf}.

    \case{$X\in^+ X$} the goal is
    \[
    \cdot \vdash I : (X_+ \to X_-) \to (X_+ \to X_-)
    \]
    \noindent which is derivable.

    \case{$X\in^p Y$} the goal is
    \[
    \cdot \vdash K\ I : (X_+ \to X_-) \to (Y \to Y)
    \]
    \noindent which is derivable.

    \case{$X \in^p (R_1\to R_2)$} let $\Gamma$ be the context
    \[
    f : (X_+ \to X_-), a : [X_p/X](R_1 \to R_2) , x : [X_{\bar{p}}/X] R_1
    \]
    \noindent Using the typing rules of System F, it suffices to show 
    \[
    \Gamma \vdash \textit{fmap}_{X,R_2}\ f\ (a\ (\textit{fmap}_{X,R_1}\ f\ x)) : [X_{\bar{p}}/X] R_2
    \]
    By the IH, since $X\in^{\bar{p}} R_1$, we have
    \[
    \cdot \vdash \textit{fmap}_{X,R_1} : (X_+ \to X_-) \to ([X_{\bar{p}}/X]R_1 \to [X_p/X]R_1
    \]
    \noindent Hence we may derive
    \[
    \cdot \vdash (\textit{fmap}_{X,R_2}\ f\ x) : [X_p/X] R_1
    \]
    \noindent and then $\cdot \vdash a\ (\textit{fmap}_{X,R_2}\ f\ x) : [X_p/X] R_2$.
     From this, using the IH with $X\in^p R_2$, we obtain the
    desired goal.

    \case{$X\in^p\all{Y}{R}$} by the IH we have
    \[
    \cdot \vdash \textit{fmap}_{X,R} : (X_+ \to X_-) \to [X_p/X]R \to [X_{\bar{p}}/X]R
    \]
    \noindent From this we obtain
    \[
    f : (X_+\to X_-) \vdash \textit{fmap}_{X,R}\ f : [X_p/X]R \to [X_{\bar{p}}/X]R
    \]
    \noindent Applying $\forall$-introduction, we get
    \[
    f : (X_+\to X_-) \vdash \textit{fmap}_{X,R}\ f : \all{Y}{[X_p/X]R \to [X_{\bar{p}}/X]R}
    \]
    \noindent Applying $\to$-introduction gives the desired conclusion (note we needed
    the $\eta$-expanded definition of $\textit{fmap}_{X,\all{Y}{R}}$).

  \end{proof}

\begin{definition}[Fold]
  $\textit{fold} \ := \ \lam{a}{\lam{x}{x\ a}}$
\end{definition}

\begin{lemma}[Fold]
  \label{lem:fold}
  Let $X$ be possibly free in $R$. Then in System F:
  \[
  \cdot\vdash \textit{fold} : \all{X}{(R \to X) \to \dparam \to  X}
  \]
\end{lemma}
\begin{opt}
  \begin{proof}
    Let $\Gamma$ be the context $a : R \to X, x:\dparam$.  It suffices
    to prove $\Gamma \vdash x\ a : X$.  Instantiating the type variable in $\dparam$ with $X$,
    we obtain
    \[
    \Gamma \vdash x : (R \to X) \to X
    \]
    \noindent So applying $x$ to $a$ indeed has type $X$ in context $\Gamma$.
  \end{proof}
\end{opt}

\begin{definition}
 \label{def:inf} 
  $\textit{in}_{X,R}\ :=\ \lam{x}{\lam{a}{a\ (\textit{fmap}_{X,R}\ (\textit{fold}\ a)\ x)}}$
\end{definition}

\begin{lemma}[In For $\dparam$ (System F)]
  \label{lem:inf}
  If $X \in^+ R$ ,
  then in System F we have
  \[
  \cdot \vdash \textit{in}_{X,R} : [\dparam/X]R \to \dparam
  \]
\end{lemma}
\begin{opt}
  \begin{proof}
    Let $\Gamma$ be the context $x : [\dparam/X]R, a : R \to X$.  Applying
    typing rules of System F, it suffices to show
    \[
    \Gamma \vdash a\ (\textit{fmap}_{X,R}\ (\textit{fold}\ a)\ x) : X
    \]
    \noindent This holds if $\Gamma \vdash \textit{fmap}_{X,R}\ (\textit{fold}\ a)\ x : R$.
    Using the assumption that $X\in^+ R$,
    instantiate Fmap (System F)~\rref{lem:fmapf} with $\dparam$ for $X_+$  and $X$ for $X_-$  to obtain:
    \[
    \cdot \vdash \textit{fmap}_{X,R} : (\dparam \to X) \to [\dparam/X]R \to R
    \]
    \noindent The desired typing follows
    using $\Gamma \vdash \textit{fold}\ a : \dparam \to X$, which holds by Fold~\rref{lem:fold}.
  \end{proof}
\end{opt}

\begin{lemma}[In For $\dparam$ (RelTT)]
  \label{lem:inp}
  If $X \in^+ R$ , then
  \[
  \textit{in}_{X,R} :: \interp{[\dparam/X]R \to \dparam}_\gamma
  \]
\end{lemma}
\begin{opt}
  \begin{proof}
    Apply Soundness Of System F For Closed Terms~\rref{cor:sndfc}
    to In For $\dparam$ (System F)~\rref{lem:inf}.
  \end{proof}
\end{opt}

We can prove a similar lemma about $\textit{in}_{X,R}$ and $\dind$,
but since $\dind$ is not a System F type we cannot use Soundness Of System F~\rref{thm:sndf}.
We first need:

\begin{lemma}[$\dind$ Containment]
  \label{lem:dindcont}
  \
  
  If $\textit{in}_{X,R} :: \interp{R \to X}_{\gamma[X\mapsto r]}$,
  then $\interp{\dind}_\gamma\subseteq r$.
\end{lemma}
\begin{opt}
  \begin{proof}
    Call the hypothesis of the lemma (1), and suppose also (2) $t\ \interp{\dind}_\gamma\ t'$.
    We must show $t\ [r]\ t'$.  Instantiating $X$ in $\dind$ with $r$, by Implicit Product~\rref{lem:impprod}
    (1) indeed implies $t\ [r]\ t'$.
  \end{proof}
\end{opt}

\begin{lemma}[In for $\dind$ (RelTT)]
\label{lem:ind}
  If $X \in^+ R$ ,  then
  \[
  \textit{in}_{X,R} :: \interp{[\dind/X]R \to \dind}_\gamma
\]
\end{lemma}
  \begin{proof}
    Assume (1) $t\ \interp{[\dind/X]R}_\gamma\ t'$ and show
    \[
    \textit{in}_{X,R}\ t\ \interp{\dind}_\gamma\ \textit{in}_{X,R}\ t'
    \]
    \noindent Unfolding the definition of
    $\dind$ and applying Internalized Typing~\rref{lem:reltpin} and Implicit Product~\rref{lem:impprod},
    it suffices to assume $r\in\mathcal{R}$ with 
\begin{equation}
    \textit{in}_{X,R}\,\interp{R \to X}_{\gamma[X\mapsto r]}\, \textit{in}_{X,R}
    \tag{2}
    \end{equation}
    \noindent and show
    \[
    \textit{in}_{X,R}\ t\ [r]\ \textit{in}_{X,R}\ t'
    \]
    \noindent This will follow from (2) if we can show (A) $t\ \interp{R}_{\gamma[X\mapsto r]}\ t'$.
    To derive this, first instantiate Monotonicity~\rref{lem:mono} with $\dind$ for $X_+$ and $r$ for $X_-$.
    That tells us that if (B) $\interp{\dind}_\gamma \subseteq r$, then also (applying Interpretation
    Over Substitution~\rref{lem:interpsubst})
    \[
    \interp{[\dind/X]R}_{\gamma}\subseteq\interp{R}_{\gamma[X\mapsto r]}
    \]
    \noindent This together with (1) proves (A).  And (B) follows from (2) by $\dind$ Containment~\rref{lem:dindcont}.
  \end{proof}

\subsection{Reflection}

Next, we prove a property known as \emph{reflection} (cf.~\cite{uustalu+99}).
For the specific case of natural numbers, a similar result is Proposition 14 of~\cite{wadler07}.
Recall the definitions of \textit{fold} and \textit{in} from Section~\ref{sec:in}.

\begin{definition}
  $\textit{rebuild}_{X,R} := \textit{fold}\ \textit{in}_{X,R}$
\end{definition}

\begin{lemma}[Reflection] 
\label{lem:reflection}
If $X\in^+ R$ , then
\[
\textit{rebuild}_{X,R}\ \interp{\dparam \to \dparam}_\gamma\ I
\]
\end{lemma}

\noindent Before we can prove this, we need:

\begin{lemma}[Fmap Fold]
  \label{lem:fmapfold}
  Suppose $Y\not\in\textit{FV}(R)$.
  Let $r_+ = \interp{f\cdot X}_\gamma$ and $r_- = \gamma(X)$.  If $X\in^p R$ ,
  then, letting $\gamma' = \gamma[Y\mapsto r_p, X\mapsto r_{\bar{p}}]$, we have
\[
\textit{fmap}_{X,R}\ f\ \interp{[Y/X]R \to R}_{\gamma'}\ I
\]
\end{lemma}
  \begin{proof}
    The proof is by induction on the derivation of $X\in^p R$.  
    We simplify implicitly using $\beta\eta$-Closure~\rref{lem:closedr}.

\case{$X\in^+ X$} since $\textit{fmap}_{X,X} = I$, the goal becomes
\[
I\ f\ [r_+ \to r_-]\ I
\]
\noindent So assume $t_1\ [r_+]\ t_2$, which is equivalent (by Deapplication~\rref{lem:deapp}) to
(1) $f\ t_1\ [\gamma(X)]\ t_2$; and show 
\[
I\ f\ t_1\ [\gamma(X)]\ t_2
\]
\noindent but this simplifies to (1). 

\case{$X\in^p Z$} since $\textit{fmap}_{X,Z} = K\ I$, the goal becomes
\[
K\ I\ f\ \interp{Z \to Z}_{\gamma'}\ I
\]
\noindent Further simplifying, it becomes
\[
I\ \interp{Z \to Z}_{\gamma'}\ I
\]
\noindent which holds obviously (Identity~\rref{lem:id}).  Since $Y\not\in\textit{FV}(R)$ by assumption,
this concludes the variable cases.

\case{$X \in^p (R_1\to R_2)$} the goal becomes
\[
\lam{a}{(\textit{fmap}_{X,R_2}\ f) \circ a \circ (\textit{fmap}_{X,R_1}\ f)}\ 
 \interp{[Y/X]R \to R}_{\gamma'}\ I
\]
\noindent So assume (1) $a\,\interp{[Y/X](R_1\to R_2)}_{\gamma'}\ a'$, and show
\[
(\textit{fmap}_{X,R_2}\ f) \circ a \circ (\textit{fmap}_{X,R_1}\ f)\ 
 \interp{R}_{\gamma'}\ a'
\]
\noindent Next, assume (2) $b\,\interp{R_1}_{\gamma'}\,b'$, and show
\[
\textit{fmap}_{X,R_2}\ f\ (a\ (\textit{fmap}_{X,R_1}\ f\ b)\ 
 \interp{R_2}_{\gamma'}\ a'\ b'
\]
\noindent Since $X\in^p R_2$, this follows by the IH from
\[
a\ (\textit{fmap}_{X,R_1}\ f\ b)\ \interp{[Y/X]R_2}_{\gamma'}\ a'\ b'
\]
\noindent In turn, this follows by (1) from
\[
\textit{fmap}_{X,R_1}\ f\ b \interp{[Y/X]R_1}_{\gamma'}\ b'
\]
\noindent Since $X\in^{\bar{p}} R_1$, this follows by the IH from (2).

\case{$X\in^+ \all{Z}{R'}$} the goal becomes
\[
\textit{fmap}_{X,R'}\ f\ \interp{[Y/X]R \to R}_{\gamma'}\ I
\]
\noindent So assume (1) $a\, \interp{\all{Z}{[Y/X]R'}}_{\gamma'}\, a'$,
and show
\[
\textit{fmap}_{X,R'}\ f\ a\ \interp{\all{Z}{R'}}_{\gamma'}\ a'
\]
\noindent For this, assume $r'\in\mathcal{R}$, and show
\[
\textit{fmap}_{X,R'}\ f\ a\ \interp{R'}_{\gamma'[Z\mapsto r']}\ a'
\]
\noindent Since $X\in^p R'$, this follows by the IH from
\[
a\, \interp{[Y/X]R'}_{\gamma'[Z\mapsto r']}\, a'
\]
\noindent But this follows by instantiating (1) with $r'$.

  \end{proof}

\noindent We may now return to:

\begin{proof}[Proof of Reflection~\rref{lem:reflection}]
    Assuming (1) $t\ \interp{\dparam}_\gamma\ t'$, it suffices
    (applying $\beta\eta$-Closure~\rref{lem:closedr}) to show
    \[
    t\ \textit{in}_{X,R}\ \interp{\dparam}_\gamma\ t'
    \]
    \noindent For this, assume $r\in\mathcal{R}$ and (2) $a\ \interp{R \to X}_{\gamma[X\mapsto r]}\ a'$,
    and show 
\begin{equation}
    t\ \textit{in}_{X,R}\ a\ [r]\ t'\ a'
\tag{A}
\end{equation}
\noindent The key idea (generalizing Wadler's Proposition 14 already mentioned)
is to instantiate (1) with the asymmetric relation
\[
\interp{\textit{fold}\ a \cdot X}_{[X\mapsto r]}
\]
\noindent Let us call this $r_a$.  (A) will follow from that instantiation if we can prove
\[
\textit{in}_{X,R}\ \interp{R \to X}_{\gamma[X\mapsto r_a]}\ a'
\]
\noindent So assume (3) $t_1\ \interp{R}_{\gamma[X\mapsto r_a]}\ t_2$, and show
\[
\textit{in}_{X,R}\ t_1\ [r_a] \ a'\ t_2
\]
\noindent This follows, by Deapplication~\rref{lem:deapp} and $\beta\eta$-Closure~\rref{lem:closedr},
from
\[
\textit{in}_{X,R}\ t_1\ a\ [r] \ a'\ t_2
\]
\noindent Further applying $\beta\eta$-Closure~\rref{lem:closedr}, this follows from
\[
a\ (\textit{fmap}_{X,R}\ (\textit{fold}\ a)\ t_1)\ [r]\ a'\ t_2
\]
\noindent By (2), this follows from 
\[
(\textit{fmap}_{X,R}\ (\textit{fold}\ a)\ t_1)\ \interp{R}_{\gamma[X\mapsto r]}\ t_2
\]
\noindent which follows from (3) by Fmap Fold~\rref{lem:fmapfold}, applying
also Environment Extension~\rref{lem:envext} to get the contexts and types
in the required form; and using $X\in^+ R$.

\end{proof}

\subsection{Equating $\dparam$ and $\dind$}

\begin{theorem}[Inductive Types]
  \label{thm:ind}

  Suppose $\textit{FV}(R) = \{ X \}$ and $X\in^+ R$ .

  \begin{enumerate}
  \item[i.] $t\,\interp{\dind\subseteq\dparam}_\gamma\,t'$
    \item[ii.] If $R$ is $\forall^+$, then $t\,\interp{\dparam\subseteq\dind}_\gamma\,t'$    
      \item[iii.] If $R$ is $\forall^+$, then  $t\,\interp{\dind \releq \dparam}_\gamma\ t'$
  \end{enumerate}
  
\end{theorem}
\begin{proof}
  Recall the definitions:
  \[
  \begin{array}{lll}
  \dparam & := & \all{X}{(R \to X) \to X} \\
  \dind & := & \all{X}{([\textit{in}_{X,R}]\,(R \to X)\,[\textit{in}_{X,R}]) \To X}
  \end{array}
  \]
  \noindent For this proof, let us apply Subset~\rref{lem:subset} implicitly.  (iii) follows from
  (i) and (ii). To show (i), assume
  $t\,\interp{\dind}_\gamma\, t'$, and instantiate $X$ in this
  assumption with $\dparam$.  This
  implies the required $t\,\interp{\dparam}_\gamma\,t'$, as long
  as (applying Interpretation Over Substitution~\rref{lem:interpsubst})
  \[
  \textit{in}_{X,R}\,\interp{[\dparam/X]R \to \dparam}_\gamma\ \textit{in}_{X,R}
  \]
  \noindent But this is exactly In For $\dparam$~\rref{lem:inp}.

  To show (ii), assume (1) $t\,\interp{\dparam}_\gamma\, t'$, and instantiate $X$ in this assumption
  with $\dind$ to get 
  \[
  t\,\interp{([\dind/X]R\to \dind) \to \dind}_\gamma\ t'
  \]
  \noindent (Here we again applied Interpretation Over Substitution~\rref{lem:interpsubst}.)
  From this and In For $\dind$~\rref{lem:ind}, we obtain (2)
  \[
  t\,\textit{in}_{X,R}\ \interp{\dind}_\gamma\ t'\ \textit{in}_{X,R}
  \]
  \noindent This is close to what we want. Applying Reflection~\rref{lem:reflection} to (1),
  we obtain
  \[
    t\,\textit{in}_{X,R}\ \interp{\dparam}_\gamma\ t' 
  \]
  \noindent Since $\textit{FV}(R) = X$, $\dparam$ is closed, so we may change $\gamma$ to $e$ here and in (1),
  by Environment Extension~\rref{lem:envext}.  Then since $R$ is $\forall^+$, $\dparam$ is also, and
  we can apply Identity Inclusion~\rref{thm:posneg} to get:
  \[
  \begin{array}{l}
    t\,\textit{in}_{X,R} =_{\beta\eta} t' \\
    t =_{\beta\eta} t'
  \end{array}
  \]
  \noindent Using these facts with $\beta\eta$-Closure~\rref{lem:closedr}, we may
  simplify (2) to the desired $t\ \interp{\dind}_\gamma\ t'$.

\end{proof}

In light of this result, we denote $\dparam$ for particular $X$ and $R$ as $D_{X,R}$,
and freely change between it and $\dind$ as long as $R$ is $\forall^+$.

\subsection{Example: \textit{Nat}}

In this section, we consider the basic example of natural numbers.  To express
this type using the parameter $R$ of $D_{X,R}$, we first need some standard
types (namely $A + B$ and $1$) and associated term definitions: for $A + B$, constructors \textit{inl} and \textit{inr},
and eliminator $\langle n , m \rangle$; and for $1$, constructor \textit{unit}.

\begin{definition}
\[
\begin{array}{lll}
  A + B & := & \all{X}{(A \to X) \to (B \to X) \to X} \\
  1 & := & \all{X}{X \to X} \\
  \textit{inl} & := & \lam{a}{\lam{x}{\lam{y}{x\ a}}} \\
  \textit{inr} & := & \lam{b}{\lam{x}{\lam{y}{y\ b}}} \\
  \langle n , m \rangle & := & \lam{c}{c\ n\ m} \\
  \textit{unit} & := & I
\end{array}
\]
\end{definition}

Now we define \textit{Nat} and its constructors as expected, with addition as an example operation:

\begin{definition}
  \[
  \begin{array}{lll}
    \textit{Nat} & := & D_{X,1 + X} \\
    \textit{zero} & := & \textit{in}_{X,1 + X}\ (\textit{inl}\ \textit{unit}) \\
    \textit{succ} & := & \textit{in}_{X,1 + X} \circ \textit{inr} \\
    \textit{add} & := & \lam{n}{\lam{m}{n\ \langle m , \textit{succ}\rangle}}
  \end{array}
  \]
\end{definition}

Thanks to Soundness of System F For Closed Terms~\rref{cor:sndfc} and the usual
System F typings of the above term definitions (including In For $\dparam$ (System F)~\rref{lem:inf}), we
have the following relational typings:

\begin{lemma}[\textit{Nat} Operations]
  \label{lem:natops}
  \[
  \begin{array}{lll}
    \textit{zero} & :: & \interp{\textit{Nat}}_\gamma \\
    \textit{succ} & :: & \interp{\textit{Nat} \to \textit{Nat}}_\gamma \\
    \textit{add} & :: & \interp{\textit{Nat} \to \textit{Nat} \to \textit{Nat}}_\gamma
  \end{array}
  \]
\end{lemma}

Following a very similar development as for Inductive Types~\rref{thm:ind}, we may also equate
$A + B$ and $1$ with inductive variants:
\begin{definition}
  \[
  \begin{array}{lll}
    A +_i B & := & \all{X}{[\textit{inl}]\,(A \to X)\,[\textit{inl}] \To\,} \\
    \ &\ & \ \ \ \ \ \ [\textit{inr}]\,(B \to X)\,[\textit{inr}]\To X  \\
    1_i & := & \all{X}{[\textit{unit}]X[\textit{unit}]\To X}
  \end{array}
  \]
\end{definition}

Recall the notation $R \releq R'$ (Definition~\ref{def:releq}).  

\begin{proposition}
  \[
  \begin{array}{l}
    t_1\ \interp{A + B \releq A +_i B}_\gamma\ t_2 \\
    t_1\ \interp{1 \releq 1_i}_\gamma\ t_2
  \end{array}
  \]
\end{proposition}

Finally, let us prove a basic inductive property of \textit{add}, as
an example.

\begin{lemma}
  \[
  \lam{n}{\textit{add}\ n\ \textit{zero}}\ \interp{\textit{Nat}\to\textit{Nat}}_\gamma\ I
  \]
\end{lemma}
\begin{proof}
  For (i): Assume (1) $n\ \interp{\textit{Nat}}_\gamma\ n'$, and show 
\begin{equation}
  \textit{add}\ n\ \textit{zero}\ \interp{\textit{Nat}}_\gamma\ n'
\tag{A}
\end{equation}
\noindent Applying Inductive Types~\rref{thm:ind} to (1) allows us
to reason inductively; we instantiate the type variable $X$ in
$\dind$ with the interpretation of
\[
r := \lam{n}{\textit{add}\ n\ \textit{zero}}\cdot \textit{Nat}
\]
\noindent We must show this is preserved by $\textit{in}_{X,1 + X}$; that is
\begin{equation}
\textit{in}_{X,1+X} :: \interp{(1 + r) \to r}_\gamma
\tag{B}
\end{equation}
\noindent By  Deapplication~\rref{lem:deapp} this suffices for (A).  For (B), assume (2) $v\ \interp{1 + r}_\gamma\ v'$,
and show 
\[
\textit{in}_{X,1+X} v \ \interp{r}_\gamma\ \textit{in}_{X,1+X} v'
\]
\noindent Switch to the inductive view of $1 + r$ in (2), and induct using the interpretation
of
\[
r' := \textit{in}_{X,1+X} \dcdot r
\]
\noindent By Deapplication~\rref{lem:deapp}, this is sufficient for (B).  We must
prove
\begin{itemize}
\item $\textit{inl}\ \textit{unit} :: \interp{r'}_\gamma$
\item $\textit{inr} :: \interp{r' \to r'}_\gamma$
\end{itemize}
\noindent Unfolding definitions of $r'$ and $r$ using Deapplication~\rref{lem:deapp}, we confirm the following
using $\beta\eta$-Closure~\rref{lem:closedr} and \textit{Nat} Operations~\rref{lem:natops}
\begin{itemize}
\item $\textit{add}\ (\textit{in}_{X,1+X}\ (\textit{inl}\ \textit{unit}))\ \textit{zero}\ \interp{\textit{Nat}}_\gamma\ (\textit{in}_{X,1+X}\ (\textit{inl}\ \textit{unit}))$
\item $\textit{add}\ (\textit{in}_{X,1+X}\ (\textit{inr}\ x))\ \textit{zero}\ \interp{\textit{Nat}}_\gamma\ (\textit{in}_{X,1+X}\ (\textit{inr}\ x'))$ from $\textit{add}\ x\ \textit{zero} \interp{\textit{Nat}}_\gamma\ x'$
\end{itemize}

\end{proof}

\subsection{Discussion}

Wadler proves a result similar to Inductive Types~\rref{thm:ind} for
the special case of the natural numbers, in Section
5 of~\cite{wadler07}.  He shows, as a theorem of a second-order logic,
that being related by the relational interpretation of
$\textit{Nat}_{\textit{param}}$ is the same as being equal natural
numbers that satisfy a predicate of unary induction.  The result here
is more general, covering any inductive datatype defined by a positive
type scheme $R$.  The equivalence is expressed not in a
second-order logic, but in RelTT.  So the proof is in terms only of
binary relations, including a binary-relational form of induction
(instead of using unary induction).  Another technical difference is
that the proof here relies on Identity Inclusion~\rref{thm:posneg}.
This does not show up in Wadler's
proof, but only because he considers just the simple example of natural numbers,
with the type $\all{X}{(X \to X) \to X \to X}$.  One may confirm
that a categorical version, as we consider here, would require
an analogous property for the proof of his Proposition 14~\cite{wadler07}.

Thanks to Inductive Types~\rref{thm:ind}, we can transport properties
between the denotations of $\dind$ and $\dparam$.  For a simple example:

\begin{lemma}
  Suppose $R$ is $\forall^+$.  Then $\interp{\dind}$ is a per.
\end{lemma}
\begin{proof}
  If $R$ is $\forall^+$, then so is $\dparam$, and hence $\interp{\dparam}_\gamma$ is a per by $\forall^+$ Per~\rref{cor:forallper}.
  This implies $\dind$ is also a per, by Inductive Types~\rref{thm:ind}.
\end{proof}

Proving this lemma directly is not hard, but using Inductive
Types~\rref{thm:ind}, unnecessary.  Richer examples are enabled thanks
to Substitutivity Of Relational Equality~\rref{lem:substreleq}.

\section{Positive-recursive types}
\label{sec:rectp}

A very useful type form from standard type theory is the recursive
type $\rec{X}{R}$, where $X$ is bound in $R$, and $X$ occurs only
positively in $R$.  The type should be isomorphic to its unfolding
$[\rec{X}{R}/X]R$, where we desire that the functions witnessing the
isomorphism are identity functions.  (This form of recursive type can
be seen as unifying the standardly distinguished \emph{isorecursive}
and \emph{equirecursive}.) This section shows how a relational version
of this type can be derived in RelTT.  The development is a
(nontrivial) adaptation of ideas from~\cite{jenkins20b}, to our
relational setting.  It is built on the derivations of subset type and
implicit product from Section~\ref{sec:basic}, and makes crucial use
of Montonicity~\rref{lem:mono}.  Let us assume that type $R$ may
contain type variable $X$ free.

\begin{definition}
\label{def:rec}
  $\rec{X}{R} := \all{X}{(R \subseteq X) \To X}$
\end{definition}

\begin{lemma}[\textit{Rec} Body]
\label{lem:rbody}
  If $\interp{R}_{\gamma[X\mapsto r]} \subseteq r$, then 
  $\interp{\rec{X}{R}}_\gamma\subseteq r$.
\end{lemma}
  \begin{proof}
    Assume (1) $t_1\ \interp{\rec{X}{R}}_\gamma\ t_2$, and
    instantiate this with $r$, to obtain
    \[
    t_1\ \interp{(R \subseteq X) \To X}_{\gamma[X\mapsto r]}\ t_2
    \]
    \noindent From this, applying Subset~\rref{lem:subset} and Implicit Product~\rref{lem:impprod},
    we have the desired
    $t_1\ r\ t_2$, as long as $\interp{R}_{\gamma[X\mapsto r]} \subseteq r$.
    But the latter is a condition of the lemma.
  \end{proof}

\begin{lemma}[\textit{Rec} Fold]
  \label{lem:rfold}
  If $X\in^+ R$, then
  $t_1\ \interp{[\rec{X}{R}/X]R \subseteq \rec{X}{R}}_\gamma\ t_2$.
\end{lemma}
  \begin{proof}
    By Subset~\rref{lem:subset},
    it suffices to show $\interp{\rec{X}{R}/X]R}_\gamma \subseteq \interp{\rec{X}{R}}_\gamma$.
So assume (1) $t\ \interp{[\rec{X}{R}/X]R}_\gamma\ t'$, and show $t\ \interp{\rec{X}{R}}_\gamma\ t'$.
Applying the semantics, Implicit Product~\rref{lem:impprod}, and Subset~\rref{lem:subset},
it suffices to assume $r\in\mathcal{R}$ and (2) $\interp{R}_{\gamma[X\mapsto r]} \subseteq r$,
and show $t_1\ [r]\ t_2$.  Applying Interpretation Over Substitution~\rref{lem:interpsubst} to (1), we have (3)
$t\ \interp{R}_{\gamma[X\mapsto \interp{\rec{X}{R}}_\gamma]}\ t'$.  By \textit{Rec} Body~\rref{lem:rbody} with (2),
$\interp{\rec{X}{R}}_\gamma \subseteq r$.  By Monotonicity~\rref{lem:mono}, (3) implies $t \interp{R}_{\gamma[X\mapsto r_\gamma]}\ t'$.
Combining this with (2), we obtain the desired $t\ [r]\ t'$.

  \end{proof}

\begin{lemma}[\textit{Rec} Unfold]
\label{lem:runfold}
  If $X\in^+ R$, then
  $t_1\ \interp{\rec{X}{R} \subseteq [\rec{X}{R}/X]R}_\gamma\ t_2$.
\end{lemma}
  \begin{proof}
    By Subset~\rref{lem:subset},
    it suffices to show $\interp{\rec{X}{R}}_\gamma \subseteq \interp{\rec{X}{R}/X]R}_\gamma$.
    So assume (1) $t\ \interp{\rec{X}{R}}_\gamma\ t'$ and show $t\ \interp{[\rec{X}{R}/X]R}_\gamma\ t'$.
    Instantiate (1) with $\interp{[\rec{X}{R}/X]R}_\gamma$ to obtain
    \[
    t\ \interp{(R \subseteq X) \To X}_{\gamma[X\mapsto\interp{[\rec{X}{R}/X]R}_\gamma]}\ t'
    \]
    \noindent Applying Interpretation Over Substitution~\rref{lem:interpsubst}, this is equivalent to
    \[
\begin{array}{l}
  t\ \interpl([[\rec{X}{R}/X]R/X]R \subseteq [\rec{X}{R}/X]R) \To\,\\
  \hspace{5.5cm} [\rec{X}{R}/X]R\interpr_\gamma\ t'
\end{array}
    \]
    \noindent By Implicit Product~\rref{lem:impprod} and Subset~\rref{lem:subset}, this implies the desired typing as long as
    \[
    \interp{[[\rec{X}{R}/X]R/X]R}_\gamma \subseteq \interp{[\rec{X}{R}/X]R}_\gamma
    \]
    \noindent But this follows by Monotonicity~\rref{lem:mono} (since $X\in^+ R$) from
    \[
    \interp{[\rec{X}{R}/X]R}_\gamma \subseteq \interp{\rec{X}{R}}_\gamma
    \]
    \noindent And this follows (by Subset~\rref{lem:subset}) directly from \textit{Rec} Fold~\rref{lem:rfold}.
  \end{proof}

\begin{theorem}[Recursive Types]
\label{thm:rectp}
  If $X\in^+ R$, then $t_1\ \interp{\rec{X}{R} \releq [\rec{X}{R}/X]R}_\gamma\ t_2$
\end{theorem}
  \begin{proof}
    Using Relational Equality~\rref{lem:releq}, this follows from \textit{Rec} Fold~\rref{lem:rfold}
    and \textit{Rec} Unfold~\rref{lem:runfold}
  \end{proof}

\section{A relational type system}
\label{sec:relty}

Having considered now some of the expressive power of RelTT, in its
ability to derive types which are often taken as primitive -- for
example, inductive types are derived here, but primitive for the
Calculus of Inductive Constructions~\cite{werner94} -- let us turn to
the question of an implementable type system for RelTT.  We follow the
approach suggested by the Curry-Howard correspondence, to to devise a
system of proof terms for derivations in RelPf.

Figure~\ref{fig:reltypfs} gives the syntax for contexts $\Gamma$ and
proof terms $p$ of RelTy, together with an erasure function mapping
these back to pure $\lambda$-calculus.  Proof terms $(p,p')$ and
$\pi\, p - x.u.v.p'$ are used for composition; the $\pi$-term is like
an existential elimination.  Erasure will indeed treat proofs of
relational typings by compositions as pairs
(Definition~\ref{def:pairs}).  The typing rules for RelTy are given in
Figure~\ref{fig:relty}.

Given a context $\Gamma$ and a proof term $p$, the rules may be read
bottom-up as an algorithm to compute the relational typing
$t\,[T]\,t'$ (if any) proved by the proof term.  Proofs
are organized in natural-deduction style: each type construct has
introduction and elimination forms.  For example, the introduction
form for an identity $t\,[t']\,t'\,t$ is $\iota\{t,t'\}$.  The elimination
is more complicated, unfortunately, as we must describe
substitution, using a proven identity $t\,[t'']\,t'$, into the
terms in some other relational typing.  The syntax for the elimination form uses a ``guide'' $\{x.t_1,t_2\}$
to give a mechanism for locating instances of $t$ in the left and right terms of the relational typing,
to be rewritten to $t'$.  The variable $x$ in terms $t_1$ and $t_2$ marks these locations.

By design, RelTy exactly follows the structure of RelPf.  Define $\ctxtdn{\Gamma}$ by
\[
\begin{array}{lll}
  \ctxtdn{\cdot} & = & \cdot \\
  \ctxtdn{\Gamma,x : t\,[R]\,t'} & = & \ctxtdn{\Gamma}, t\,[R]\,t'
\end{array}
\]
\noindent This maps RelTy contexts to RelPf contexts.  A reverse mapping $\ctxtmv{\Gamma}$ can be defined
as $\ctxtmv{\Gamma}_k$ where $k$ is the length of $\Gamma$, and the helper function is defined as follows, using a canonical ordering
$x_1,x_2,\ldots$ for assumption variables:
\[
\begin{array}{lll}
  \ctxtmv{\cdot}_k & = & \cdot \\
  \ctxtmv{\Gamma,t\,[R]\,t'}_k & = & \ctxtmv{\Gamma}_{k-1}, x_k : t\,[R]\,t'
\end{array}
\]

\begin{theorem}[RelTy-RelPf Isomorphism]
  \label{thm:relty}
\ 

\begin{enumerate}
  \item[i.] If $\Gamma\vdash p : t\,[R]\,t'$ in RelTy, then $\ctxtdn{\Gamma} \vdash t\,[R]\,t'$ in RelPf.
  \item[ii.] If $\Gamma\vdash t\,[R]\,t'$ in RelPf, then there exists $p$ such that $\ctxtmv{\Gamma} \vdash p : t\,[R]\,t'$ in RelTy.
\end{enumerate}
\end{theorem}
\begin{proof}
  For (i): because RelTy just expands RelPf with proof terms, the proof amounts to erasing all proof terms
  (including assumptions $u$ in contexts) from RelTy derivations.  For (ii): by design, RelTy has proof-term constructs
  corresponding to all proof rules of RelPf, so the proof amounts to recursively adding in those terms.
\end{proof}

If we project even further, we can map from RelTy to System F.  Recall
the definition of pairs (Definition~\ref{def:pairs}), which are used
in projecting composition.
\begin{definition}
  Define $|R|$ recursively by:
  \[
  \begin{array}{lll}
    |X| & = & X \\
    |R \to R'| & = & |R|\to|R'| \\
    |\all{X}{R}| & = & \all{X}{|R|} \\
    |R^\cup| & = & |R| \\
    |R \cdot R'| & = & |R|\times|R'| \\
    |t| & = & \all{X}{X\to X}
  \end{array}
  \]

  Extend this to contexts by recursively defining $|\Gamma|$:
\[
\begin{array}{lll}
  |\cdot| & = & \cdot \\
  |\Gamma,u:t\,[R]\,t'| & = & |\Gamma|, u : |R|
\end{array}
\]
  \end{definition}

\begin{theorem}[RelTy Projection]
\label{thm:reltyproj}
  If $\Gamma \vdash p : t\,[R]\,t'$ then $|\Gamma|\vdash |p| : |R|$ in System F.
  \end{theorem}
\begin{opt}
\begin{proof} The proof is by induction on the assumed RelTy derivation.

\caseb
\[
    \infer{\Gamma\vdash x : t\,[R]\,t'}{x : t\,[R]\,t' \in \Gamma}
\]
\noindent From $x:t\,[R]\,t'\in \Gamma$ we get $x:|R|\in\Gamma$ and hence the desired conclusion.

\caseb
\[
    \infer{\Gamma\vdash \lam{u:R}{p} : \lam{x}{t}\,[R \to R']\,\lam{x'}{t'}}
          {\Gamma, u : x\,[R]\,x' \vdash p : t\,[R']\,t' & (*)}
\]
\noindent By the IH we have $|\Gamma|,u:|R|\vdash p:|R'|$, from which we deduce
the desired $|\Gamma|\vdash\lam{u}{|p|}:|R\to R'|$.

\caseb
\[
 \infer{\Gamma\vdash p_1\,p_2 : t_1\,t_2\,[R']\,t_1'\,t_2'}
          {\Gamma \vdash p_1 : t_1\,[R \to R']\,t_1'
            & \Gamma \vdash p_2 : t_2\,[R]\,t_2'} 
\]
\noindent By the IH we have $|\Gamma|\vdash|p_1| : |R \to R'|$ and $|\Gamma|\vdash|p_2| : |R|$,
from which we deduce the desired $|\Gamma|\vdash|p_1\,p_2| : |R'|$.

\caseb
\[
         \infer{\Gamma\vdash p\{R\} : t\,[[R/X]R']\,t'}
                {\Gamma \vdash p : t\,[\all{X}{R'}]\,t'}
\]
\noindent By the IH we have $|\Gamma|\vdash|p|:\all{X}|R'|$, from which the desired $|\Gamma|\vdash|p|:|[R/X]R'|$ follows.

\caseb
\[
          \infer{\Gamma\vdash \Lam{X}{p} : t\,[\all{X}{R}]\,t'}
                {\Gamma \vdash p : t\,[R]\,t' & X\not\in\textit{FV}(\Gamma)}
\]
\noindent By the IH we have $|\Gamma|\vdash|p| : |R|$, from which the desired $|\Gamma|\vdash|p|:\all{X}{|R|}$ follows.

\caseb
\[
          \infer{\Gamma\vdash  t_1' \convl p \convr t_2' : t_1'\,[R]\,t_2'}
                {\Gamma \vdash p : t_1\,[R]\,t_2 & t_1 =_{\beta\eta} t_1' & t_2 =_{\beta\eta} t_2'} 
\]
\noindent The erasure of $t_1'\convl p \convr t_2'$ is $|p|$, so the desired conclusion is just $|\Gamma|\vdash|p|:|R|$,
which we have by the IH.

\caseb
\[
         \infer{\Gamma \vdash \cup_e\, p : t'\,[R]\,t}
               {\Gamma \vdash p : t\,[R^\cup]\,t'}
\]
\noindent Similar to the previous case.

\caseb
\[
        \infer{\Gamma \vdash \cup_i\, p : t'\,[R^\cup]\,t}           
               {\Gamma \vdash p : t\,[R]\,t'}
\]
\noindent Similar to the previous case.

\caseb
\[
         \infer{\Gamma \vdash \iota\{t,t'\} : t\,[t']\, t'\, t}{\ }
\]
\noindent $|\iota\{t,t'\}|$ is $I$ and erasure of the term promotion $t'$ is $\all{X}{X\to X}$.  So this inference
translates to the familiar typing of the identity function in System F.

\caseb
\[
         \infer{\Gamma \vdash \rho \{ x.t_1,t_2\}\ p - \ p' : [t'/x]t_1\,[R]\,[t'/x]t_2}
               {\Gamma \vdash p : t\,[t'']\,t' & \Gamma \vdash p' : [t''\,t/x]t_1\,[R]\,[t''\,t/x]t_2}
\]
\noindent By the IH we have $|\Gamma|\vdash|p'|:|R|$.  Since the $\rho$-proof erases to just the erasure of its leftmost
subproof, this suffices for the desired conclusion.

\caseb
\[
\infer{\Gamma \vdash \pi\, p - x.u.v.p' : t_1\,[R'']\,t_2}
              {
                \begin{array}{c}
                  \Gamma \vdash p : t\,[R\cdot R']\,t'
                  \quad (**)
                  \\ \Gamma , u : t\,[R]\,x, v : x\,[R']\,t'\vdash p' : t_1\,[R'']\,t_2
                \end{array}
              }
\]
\noindent By the IH we have $|\Gamma|\vdash|p|:|R|\times|R'|$ and $|\Gamma|,u:|R|,v:|R'|\vdash|p'|:|R''|$.  By the definition
of product types in System F, from these derivations we may easily establish $|\Gamma|\vdash|p|\,\lam{u}{\lam{v}{|p'|}} : |R''|$,
which suffices since $|\pi\, p - x.u.v.p'| = |p|\ \lam{u}{\lam{v}{|p'|}}$.

\caseb
\[
        \infer{\Gamma \vdash (p,p') : t\,[R\cdot R']\,t'}
               {\Gamma \vdash p : t\,[R]\,t'' & \Gamma \vdash p':t''\,[R']\, t'}
\]
\noindent By the IH we have $|\Gamma|\vdash|p|:|R|$ and $|\Gamma|\vdash|p'|:|R'|$.  With these
we may deduce $|\Gamma|\vdash(|p|,|p'|):|R|\times|R'|$ by the definition of product types in System F.
  
\end{proof}
\end{opt}

This result is interesting, because it shows that any valid RelTy proof term proves a property of its own erasure:
\begin{proposition}[RelTy Self]
  \label{prop:reltyself}
  If $\Gamma \vdash p : t \,[R]\,t'$, then $\Gamma\vdash p : |p|\,[R]\,|p|$.
\end{proposition}
\begin{proof}[Proof sketch]
  From the assumed RelTy derivation we get to $\Gamma\vdash
  |p|\,[R]\,|p|$ using RelTy Projection~\rref{thm:reltyproj} and
  Soundness of System F~\rref{thm:sndf}.  We need then just a somewhat
  more informative version of part (ii) of RelTy-RelPf
  Isomorphism~\rref{thm:relty}, which maps RelPf derivations to
  particular proof terms $p$ (not just showing that some such $p$
  exists) in correspondence with the RelPf derivations.
  \end{proof}

\begin{figure}
  \[
  \begin{array}{lll}
\Gamma & ::= & \cdot\ |\ \Gamma,u:t\,[R]\,t' \\ \\
    \textit{proof terms } p & := & u\ |\ \lam{u:T}{p}\ |\ p\ p' \ |\\
    \ &\ & p\{T\}\ |\ \Lam{X}{p}\ |\\
    \ &\ & t\convl p \convr t'\ | \\
    \ &\ & \cup_i\, p\ |\ \cup_e\, p \ |\\
    \ &\ & \iota\{t,t'\}\ |\ \rho \{ x.t_1,t_2\}\ p - p' \ |\\
    \ &\ & (p,p')\ |\ \pi\, p - x.u.v.p'
    \\ \\
    |u| & = & u \\
    |\lam{u:T}{p}| & = & \lam{u}{p} \\
    |p\ p'| & = & |p|\ |p'| \\
    |p\{T\}| & = & |p| \\
    |\Lam{X}{p}| & = & |p| \\
    |t\convl p \convr t'| & = & |p| \\
    |\cup_i\, p| & = & |p| \\
    |\cup_e\, p| & = & |p| \\
    |\iota\{t,t'\}| & =  & I \\
    |\rho \{ x.t_1,t_2\}\ p - p'| & = & |p'| \\
    |(p,p')| & = & (|p|, |p'|) \\
    |\pi\, p - x.u.v.p'| & = & |p|\,\lam{u}{\lam{v}{|p'|}} 
  \end{array}
\]
  \caption{Syntax for proof terms of RelTy, and erasure to pure $\lambda$-calculus}
  \label{fig:reltypfs}
\end{figure}

\begin{figure*}[!t]
  \[
  \begin{array}{lll}
    \infer{\Gamma\vdash x : t\,[R]\,t'}{x : t\,[R]\,t' \in \Gamma}
    &
    \infer{\Gamma\vdash \lam{x:R}{p} : \lam{x}{t}\,[R \to R']\,\lam{x'}{t'}}
          {\Gamma, p : x\,[R]\,x' \vdash p : t\,[R']\,t' & (*)}
  & \infer{\Gamma\vdash p_1\,p_2 : t_1\,t_2\,[R']\,t_1'\,t_2'}
          {\Gamma \vdash p_1 : t_1\,[R \to R']\,t_1'
            & \Gamma \vdash p_2 : t_2\,[R]\,t_2'} 
  
\\\\

          \infer{\Gamma\vdash p\{R\} : t\,[[R/X]R']\,t'}
                {\Gamma \vdash p : t\,[\all{X}{R'}]\,t'}
                &
          \infer{\Gamma\vdash \Lam{X}{p} : t\,[\all{X}{R}]\,t'}
                {\Gamma \vdash p : t\,[R]\,t' & X\not\in\textit{FV}(\Gamma)}
                &
          \infer{\Gamma\vdash  t_1' \convl p \convr t_2' : t_1'\,[R]\,t_2'}
                {\Gamma \vdash p : t_1\,[R]\,t_2 & t_1 =_{\beta\eta} t_1' & t_2 =_{\beta\eta} t_2'} 
\\\\
\multicolumn{2}{l}{
         \infer{\Gamma \vdash \cup_e\, p : t'\,[R]\,t}
               {\Gamma \vdash p : t\,[R^\cup]\,t'}
                           
\ \ \ \ \
         \infer{\Gamma \vdash \cup_i\, p : t'\,[R^\cup]\,t}           
               {\Gamma \vdash p : t\,[R]\,t'}
                
        \ \ \  \ \ 
         \infer{\Gamma \vdash \iota\{t,t'\} : t\,[t']\, t'\, t}{\ }}
         &
         \infer{\Gamma \vdash \rho \{ x.t_1,t_2\}\ p - \ p' : [t'/x]t_1\,[R]\,[t'/x]t_2}
               {\Gamma \vdash p : t\,[t'']\,t' & \Gamma \vdash p' : [t''\,t/x]t_1\,[R]\,[t''\,t/x]t_2}
               \\\\
               
      \multicolumn{2}{l}{\infer{\Gamma \vdash \pi\, p - x.u.v.p' : t_1\,[R'']\,t_2}
              {\Gamma \vdash p : t\,[R\cdot R']\,t' & \Gamma , u : t\,[R]\,x, v : x\,[R']\,t'\vdash p' : t_1\,[R'']\,t_2 & (**)}}
   &     \infer{\Gamma \vdash (p,p') : t\,[R\cdot R']\,t'}
               {\Gamma \vdash p : t\,[R]\,t'' & \Gamma \vdash p':t''\,[R']\, t'}

  \end{array}
  \]
  \hspace{.2cm}Side condition (*) is $x\not\in\textit{FV}(\Gamma, R, R')$.

  \hspace{.2cm}Side condition (**) is $x\not\in\textit{FV}(\Gamma, t_1, t_2, t, t', R, R', R'')$.

  \caption{RelTy typing rules}
  \label{fig:relty}
\end{figure*}

\section{Related work}
\label{sec:related}

RelTT's semantics (Figure~\ref{fig:relsem}) is
a relational realizability semantics, where realizers
are terms of untyped lambda calculus (cf.~\cite{oosten02,troelstra98}).  Relational semantics for types has
been studied extensively in the context of logical relations; see
Chapter 8 of~\cite{mitchell96}.  An influential branch of this work
was initiated by Reynolds, on what is now called
parametricity~\cite{reynolds83}.  ~\cite{sojakova18} frames some recent
results, using categorical semantics.

\cite{kd13} proposes a similar realizabiliity semantics, for
the Calculus of Constructions plus an extensional equality type.  The
major difference is that in RelTT, we propose a
notation for asymmetric relations, which is lacking in~\cite{kd13}.
Instead,
constructions based on the semantics are done at the meta-level (where
asymmetric relations can be described).  Indeed, the denotable relations
of~\cite{kd13} are partial equivalences -- albeit of a modified form
due to basing the semantics on ``zig-zag complete'' relations.  In contrast,
we have seen above some families of types whose denotations are
partial equivalences (unmodified) in RelTT.  But by design, not all
types denote partial equivalences in RelTT, since reasoning about
terms generally involves asymmetric relations; an important example
we saw is Reflection~\rref{lem:reflection}.
    
Observational Type Theory (OTT) is an approach to type-specific
extensionality principles in an intensional dependent type theory,
based on a primitive heterogeneous equality type and associated
operators~\cite{altenkirch+07}.  RelTT is similar in deriving
extensionality principles, but more radical in design: where OTT
extends a traditional (i.e., unary) type theory including $W$-types
with an extensional form of equality, RelTT takes a binary view of all
types, and does not use dependent types at all.  The resulting system
is hence formally quite a bit simpler.

Unlike~\cite{reynolds83} and subsequent works like~\cite{atkey12}),
RelTT lacks Identity Extension. This property states that when free
type variables are interpreted by identity relations, the relational
meaning of a type $T$ is the identity relation on the unary (or
``object'') interpretation of $T$.  This is a very strong property,
showing that the object interpretation of types gives canonical forms
for the equivalence defined by the relational interpretation of types.
But it rules out expression of asymmetric relations as types.  RelTT
preserves this possibility, at the cost of weakening Identity
Extension to Identity Inclusion~\rref{thm:posneg}.

In~\cite{plotkin93}, Plotkin and Abadi introduce a second-order logic
for reasoning about (typable) terms of System F by quantification over
relations, and using a parametricity axiom.  In contrast, RelTT uses
relational types to express relations in a more compact way.  A
parametricity axiom would not make sense here, for there is no
separate notion of unary typing from which relational typing could be
stated to follow.  The only typings are relational.  

RelTT may be compared with previous work of Stump et al.  on
Cedille~\cite{stump20,FBS18_Efficient-Mendler,stump18}.  Both systems
aim at a minimalistic extension of a small pure type system as a
foundation for type theory.  Cedille extends the Curry-style Calculus
of Constructions with dependent intersections, implicit products, and
an equality type over untyped terms.  RelTT extends System F with
three relational operators based on a relational semantics.  While the
systems are roughly equivalent in formal complexity -- with RelTT
having the simplifying advantage of eschewing dependent types -- RelTT
delivers type-specific extensionality principles, which Cedille lacks.

\cite{Parametricity+TT} considers how parametricity results can be
embedded in constructive type theory, by elaborating types into
corresponding theorems in the logic of so-called ``reflective'' pure
type systems.  Subsequent work built an extended PTS internalizing
these theorems~\cite{Computational+Parametricity}. These papers
consider fairly rich Church-style lambda calculi, in contrast to the
more compact Curry-style calculus of RelTT.  

Finally, RelTT may be compared with Homotopy Type Theory (HoTT), a
line too active in recent years to summarize here~\cite{hottbook}.  Both
theories support functional extensionality.  The two approaches
have different origins: logical relations and parametricity for RelTT,
homotopy theory and higher category theory for HoTT.  A major point of
difference is univalence: while RelTT allows one to express and derive
relational equalities within the theory, these are based on semantic
inclusions, not isomorphisms (as in univalence).  Thus, transporting
results between isomorphic types as done in HoTT is not (in an obvious
way) directly possible in RelTT.  Another point of comparison is the
compactness of the theory.  RelTT is based on a very compact semantics
for a small number of relational type forms.  In contrast, systems
like, for one notable example, Cubical Agda, are based on a larger
array of primitives~\cite{vezzosi19}.  Whereas the free theorems
provided by parametricity allows proofs to be transported to
observationally equivalent terms, HOTT uses explicit equivalences
between terms for this purpose. Only very recent work has considered
how to combine these two complementary approaches inside of univalent
type theories~\cite{Univalence+Parametricity}.

\section{Conclusion and future work}

Based on a binary relational semantics, RelTT is a new minimalistic
extensional type theory, where inductive and positive-recursive types
are derivable.  The theory does not have dependent types, and indeed,
an indirect conclusion of the paper is that type theory does not
require dependent types for reasoning about programs.  Just passing
from the traditional unary semantics to a binary-relational one opens
the possibility for formal (extensional) reasoning about programs.
Future work includes direct support for existential types, for
deriving coinductive types; the standard double-negation encoding
of existentials is problematic due to the requirement of forall-positivity
for Identity Inclusion~\rref{thm:posneg}.


\section*{Acknowledgments}

We gratefully acknowledge NSF
  support under award 1524519, and DoD support under award
  FA9550-16-1-0082 (MURI program).  First author: St. Jer., AMDG.






\end{document}